\documentclass[pdflatex,sn-basic,Numbered]{sn-jnl}
\usepackage{braket}
\usepackage{graphicx}
\usepackage{multirow}
\usepackage{subfigure}
\usepackage{bm}
\usepackage{mathtools}
\mathtoolsset{showonlyrefs}
\usepackage{booktabs}
\usepackage{amsthm, amsmath,amssymb,amsfonts}
\usepackage{mathrsfs}
\usepackage{xcolor}
\usepackage{manyfoot}
\usepackage{float}
\DeclareMathOperator*{\argmin}{\arg\!\min}

\newcommand{\Mod}[1]{(\mathrm{mod}\, #1)}

\newcommand{\Bb}{\mathcal{B}}
\newcommand{\Cc}{\mathcal{C}}

\newcommand{\Nn}{\mathcal{N}}
\newcommand{\Ww}{\mathcal{W}}
\newcommand{\Uu}{\mathtt{U}}
\newcommand{\ef}{\boldsymbol{e}}
\newcommand{\Bf}{\boldsymbol{B}}
\newcommand{\Ef}{\boldsymbol{E}}
\newcommand{\Wf}{\boldsymbol{W}}
\newcommand{\xx}{\boldsymbol{x}}
\newcommand{\EE}{\mathbb{E}}
\newcommand{\RR}{\mathbb{R}}
\newcommand{\D}{\mathrm{d}}
\newcommand{\E}{\mathrm{e}}
\newcommand{\I}{\mathrm{i}}
\newcommand{\Id}{\mathrm{I}}
\newcommand{\ub}{\mathfrak{u}}
\newcommand{\half}{\frac{1}{2}}
\newcommand{\dt}{\partial_{t}}
\newcommand{\Oneovd}{\frac{1}{d}}
\newcommand{\Tr}{\mathrm{T}}
\newcommand{\Ug}{\mathtt{U}}
\newcommand{\Ag}{\mathtt{A}}
\newcommand{\Bg}{\mathtt{B}}
\newcommand{\Cg}{\mathtt{C}}
\newcommand{\Gg}{\mathtt{G}}
\newcommand{\Mg}{\mathtt{M}}
\newcommand{\Rg}{\mathtt{R}}
\newcommand{\Xg}{\mathtt{X}}
\newcommand{\Zg}{\mathtt{Z}}
\newcommand{\Cnotg}{\mathtt{C}_{\mathtt{NOT}}}
\theoremstyle{plain}%
\newtheorem{theorem}{Theorem}
\numberwithin{theorem}{subsection}

\newtheorem{lemma}[theorem]{Lemma}
\newtheorem{remark}[theorem]{Remark}
\newtheorem{definition}[theorem]{Definition}

\numberwithin{figure}{section}
\numberwithin{equation}{section}
\DeclareMathOperator{\acos}{\mathrm{acos}}

\usepackage[textsize=tiny]{todonotes}

\numberwithin{equation}{section}
\begin{document}

\title[Unsupervised Random Quantum Networks for PDEs]{Unsupervised Random Quantum Networks for PDEs}

\author*[1,2]{\fnm{Josh} \sur{Dees}}\email{jd1622@ic.ac.uk}

\author[2]{\fnm{Antoine} \sur{Jacquier}}\email{a.jacquier@imperial.ac.uk}

\author[1]{\fnm{Sylvain} \sur{Laizet}}\email{s.laizet@imperial.ac.uk}

\affil*[1]{\orgdiv{Department of Aeronautics}, \orgname{Imperial College London}, \orgaddress{\street{South Kensington Campus}, \city{London}, \postcode{SW7 2AZ}, \country{United Kingdom}}}

\affil[2]{\orgdiv{Department of Mathematics}, \orgname{Imperial College London}, \orgaddress{\street{South Kensington Campus}, \city{London}, \postcode{SW7 2AZ}, \country{United Kingdom}}}

\abstract{

Classical Physics-informed neural networks (PINNs) approximate solutions to PDEs with the help of deep neural networks trained to satisfy the differential operator and the relevant boundary conditions. 
We revisit this idea in the quantum computing realm, using  parameterised random quantum circuits as trial solutions.
We further adapt recent PINN-based techniques to our quantum setting, in particular Gaussian smoothing.
Our analysis concentrates on the Poisson, the Heat and the  Hamilton-Jacobi-Bellman equations, 
which are ubiquitous in most areas of science.
On the theoretical side, we develop a complexity analysis of this approach, 
and show numerically that random quantum networks can 
outperform more traditional quantum networks as well as random classical networks. 
}

\date{\today}

\maketitle
\newpage
\newpage

\section{Introduction}\label{sec1}
Partial Differential Equations (PDEs) describe continuous phenomena, such as the fluid flows or the propagation of waves. 
They correspond to the mathematical translation of observable, physical, chemical or biological processes. Unfortunately, these equations rarely admit analytical solutions, 
and need to be be discretised on some mesh. 
This process can be computationally expensive, especially for high-dimensional problems and when unstructured meshes are required, for example to account for local irregular behaviours. 
This discretised scheme can then be solved using a variety of numerical methods, 
such as finite elements (FEM), finite differences (FDM) or  the finite volume (FVM). 
However, even these methods can be inefficient for large and complex problems. 
For example, the solution of the Navier-Stokes equations, describing the motions of a fluid, can require millions of hours of CPU or GPU time on a supercomputer. 
Another example is the Poisson equation, one of the most important PDEs in engineering, 
including heat conduction, gravitation, and electrodynamics. 
Solving it numerically in high dimension is only tractable with iterative methods, 
which often do not scale well with dimension and/or require specialist knowledge 
when dealing with boundary conditions or when generating the discretisation mesh. 

Neural networks (NNs) are well positioned to solve such complicated PDEs and are already used in various areas of engineering and applied mathematics for complex regression and image-to-image translation tasks. 
The scientific computing community has applied them PDE solving as early as the 1980s~\cite{LEE1990110}, yet interest has exploded over recent years, due in part to significant improvements in computational techniques and improvements in the formulation of such networks, as detailed and highlighted for example in~\cite{lu2021deepxde,cuomo2022scientific,wenshu2022review}.

Quantum computing is a transformative new paradigm which takes advantage of the quantum phenomena seen at the microscopic physical scale. 
While significantly more challenging to design, quantum computers can run specialised algorithms that scale better than their classical counterparts, sometimes exponentially faster. 
Quantum computers are made of quantum bits (or qubits) that, unlike bits in conventional digital computers, store and process data based on two key principles of quantum physics: quantum superposition and quantum entanglement. 
They characteristically suffer from specific errors, namely quantum errors, which are related to the quantum nature of their qubits. 
Even if quantum computers of sufficient complexity are not yet available, there is a clear need to understand what tasks we can hope to perform thereon and to design methods to mitigate the effects of quantum errors~\cite{preskill2023quantum}.

Quantum neural networks form a new class of machine learning networks and leverage quantum mechanical principles such as superposition and entanglement with the potential to deal with complicated problems and/or high-dimensional spaces.
Suggested architectures for quantum neural networks include~\cite{Zoufal_2019,dunjko2018machine,gonon2023universal} 
and suggest that there might be potential advantages, including faster training. 
Preliminary theoretical research into quantum machine learning shows that quantum networks can produce a more trainable model~\cite{Abbas_2021}. 
This is particularly relevant to solving PDEs with machine learning as techniques which produce a more favourable loss landscape can drastically improve the performance of these models~\cite{krishnapriyan2021characterizing,gopakumar2023loss}.

In the present work, we suggest a new way of formulating a quantum neural network, translate some classical machine learning techniques to the quantum setting and develop a complexity analysis 
in the context of specific PDEs 
(the Heat, the Poisson and an HJB equation).
This provides a framework to demonstrate the potential and the versatility of quantum neural networks as PDE solvers. 

The paper is organised as follows: 
Section~\ref{sec:Tools} introduces the PINN algorithm and reviews the basics of classical and quantum networks.
In Section~\ref{sec:RNN}, 
we introduce a novel quantum network
to solve  specific PDEs and provide a complexity analysis.
Finally, we confirm the quality of the scheme numerically in Section~\ref{sec:Numerics}.\\

\textbf{Notations.}
We denote by $\Cc^{n}(\Omega,\RR)$ the space of~$n$ times differentiable functions from~$\Omega$ to~$\RR$, by~$L^p(\Omega)$ the space of functions with finite~$L^p$ norm and define the Sobolev space 
$\Ww^{\alpha,\beta}(\Omega): = \left\{f \in L^\beta(\Omega): \nabla^s f \in L^\beta(\Omega) \text{ for all }|s| \leq \alpha\right\}$,
where $\nabla^s f$ refers to the weak derivative of order~$s$.
Similarly, we let~$\Ww_0^{\alpha,\beta}$ be the subspace of functions in~$\Ww^{\alpha,\beta}(\Omega)$ that vanish in the trace sense on the boundary of~$\Omega$. 
We use $\|\cdot\|$ to refer to the Euclidean norm. 

\section{Main tools}\label{sec:Tools}
\subsection{PINN algorithm}\label{Pinnalgo}
The Physics-informed neural network (PINN) algorithm relies on the expressive power of neural networks to solve PDEs. 
Let $\Omega \subset \RR^{d}$ be a bounded Lipschitz domain with boundary~$\partial\Omega$,  $\mathscr{F}: \Cc^{K}(\Omega,\RR) \to \Cc(\Omega,\RR)$ a differential operator of order at most~$K$, $E \subset \partial \Omega$ and $h: E \to \RR$. 
The goal is to estimate the solution $u: \Omega \to \RR $ to the PDE
\begin{align}\label{eq:PDE}
 \begin{cases}
 \mathscr{F}\left(u\right)(x) = 0, & \text{for all } x\in \Omega, \\ 
 u(x)=h(x), & \text{for all } x\in E.
 \end{cases}
\end{align}
Let $u_\Theta : \Omega \to \RR$ be a neural network at least~$K$ times continuously differentiable parameterised by some set~$\Theta$.
We assume access to two datasets: 
independent and identically distributed (iid) random vectors $\{X^{(e)}_i\}_{i=1,\ldots,n_e}$ with known distribution~$\mu_E$ on~$E$ and iid random vectors $\{X^{(r)}_i\}_{i=1,\ldots, n_r}$ with known distribution~$\mu_\Omega$ on~$\Omega$. 
We then minimise the empirical loss function
 
\begin{equation}\label{eq:EmpLossFunction}
    \mathcal{L}_{n_e,n_r}(u_\Theta) := \frac{\lambda_e}{n_e}\sum_{i=1}^{n_e}\left|
    u_\Theta\left(X^{(e)}_i\right) - h\left(X^{(e)}_i\right)\right|
    + \frac{1}{n_r} \sum_{i = 1}^{n_r}
    \left|\mathscr{F}\left(u_\Theta \right)\left(X^{(r)}_i\right)\right|, 
\end{equation}
over the set~$\Theta$ using a hybrid quantum-classical training loop, where $\lambda_e >0$ is a hyper-parameter 
allowing one to balance the two loss components. 
This training loop evaluates all~$u_\Theta (\cdot)$ values on a quantum computer before feeding the values to a classical computer for use in classical optimisation routines. 
As shown in~\cite{doumeche2023convergence} (Proposition 3.2 and the associated discussion) minimising~\eqref{eq:EmpLossFunction} does not necessarily imply anything about the mean squared error $\lvert \lvert u - u_\Theta \rvert \rvert^2_2$. 


\subsection{Gaussian smoothing}
In~\cite{he2023learning}, the authors investigated Gaussian smoothing the output of a classical neural network for use as a PDE trial solution, 
providing a simpler expression (as an expectation) for all derivatives. 
Consider indeed a trial solution of the form
\begin{equation}\label{smoothing1}
f(x) := \EE_{\delta \sim \Nn(0,\Id\sigma^{2})}[u(x+ \delta)],
\qquad\text{for all }x\in\RR^d,
\end{equation}
for some $\sigma>0$, where~$u$ is the output of a neural network.
Then, assuming~$u$ measurable, 
all derivatives can then be written as follows:

\begin{theorem}[{\cite[Theorem 1]{he2023learning}}]\label{thm:delg}
For any measurable function $u: \RR^{d} \to \RR$, the function~$f$ defined in~\eqref{smoothing1} is differentiable and the following holds for all $x\in\RR^d$:
\begin{equation}
\nabla f(x) = \frac{1}{\sigma ^2} \EE_{\delta \sim \Nn(0,\Id\sigma^{2})} \left[ \delta u(x+ \delta)\right].
\end{equation}
\end{theorem}

Theorem~\ref{thm:delg} implies that an (unbiased) estimate for the gradient can be computed for example by Monte Carlo, for example as 
\begin{equation}
\widetilde{\nabla f}(x)
:= \frac{1}{K} \sum_{k=1}^{K} \frac{\delta_k}{\sigma ^2}u(x+ \delta_k),
\end{equation}
using~$K$ iid Gaussian $\Nn(0,\Id\sigma^{2})$ samples $(\delta_k)_{k=1,\dots,K}$. 
This can be improved using a combination of antithetic variable and control variate techniques (see for example~\cite[Chapter~4]{glasserman2004monte} for a thorough overview) resulting in the new estimator
\begin{align}\label{eq:cavdel}
\widetilde{\nabla f}(x) := \frac{1}{K} \sum_{k=1}^{K} \frac{\delta_k}{2 \sigma ^2}\Big(u(x+ \delta_k) - u(x- \delta_k)\Big).
\end{align}
This method can easily be extended to derivatives of any order with recursion, for example for the Laplacian:
\begin{equation}
\widetilde{\Delta f}(x)
:= \frac{1}{K} \sum_{k=1}^{K} \frac{\|\delta_k\| ^2 - 2\sigma ^2 }{2\sigma^4}
\Big(u(x + \delta_k) - 2u(x) + u(x-\delta_k)\Big).
\end{equation}
In fact, the function~$f$ defined in~\eqref{smoothing1} is Lipschitz continuous:
\begin{theorem}\label{thm:lip}
Let $u:\RR^d\to\RR$ be measurable.
The map
$\EE_{\delta \sim \Nn(0,\Id\sigma^{2})}[u(\cdot+ \delta)]$ is Lipschitz with respect to any arbitrary norm. 
In particular with the $2$-norm, it is\textit{} $\frac{\ub}{\sigma}\sqrt{\frac{2}{\pi}}$ Lipschitz, with $\ub = \sup_{x \in \RR^d}  |u(x)|$.
\end{theorem}
\begin{proof}
This statement is proved in~\cite{he2023learning} for the $\ell_2$-norm, and it is easy to extend to any arbitrary norm. 
Since~$f$ is differentiable by Theorem~\ref{thm:delg}, it is well known that its Lipschitz constant~$L^{\alpha}(f)$ in the $\alpha$-norm can be obtained as
$$
L^{\alpha}(f) = \sup_{x \in \RR^d} 
\|\nabla_x f(x)\|_{\alpha}^{*},
$$
with $\|z\|_{\alpha}^{*} := \sup_{\|y\|_{\alpha}\leq 1}y^\top z$ the dual norm of $\|\cdot\|_\alpha$. 
Using Theorem~\ref{thm:delg}, we can write
\begin{align}
L^{\alpha}(f)
= \sup_{x \in \RR^d}  \left\lVert \EE\left[\frac{\delta}{\sigma ^2} u(x+ \delta) \right]\right\rVert_{\alpha}^{*}
&= \frac{1}{\sigma ^2}\sup_{x \in \RR^d} \sup_{\|y\|_{\alpha}\leq 1}\left|y^\top\EE\left[\delta u(x+ \delta) \right]\right|\\
&= \frac{1}{\sigma ^2}\sup_{x \in \RR^d} \sup_{\|y\|_{\alpha}\leq 1}
\left|\EE\left[y^\top \delta u(x+ \delta) \right]\right|\\     
&\leq \frac{\ub}{\sigma ^2} \sup_{\|y\|_{\alpha}\leq 1}
\left|\EE\left[ y^\top\delta\right]\right|\\
&\leq \frac{\ub}{\sigma ^2} \sup_{\|y\|_{\alpha}\leq 1}
\EE\left[\left| y^\top\delta\right|\right]
\leq \frac{\ub}{\sigma ^2} \sup_{\|y\|_{\alpha}\leq 1}\sigma \sqrt{\frac{2\|y\|_2}{\pi}}.
\end{align}
When $\alpha = 2$, the upper bound becomes $\frac{\ub}{\sigma}  \sqrt{\frac{2}{\pi}}$ and the theorem follows.
\end{proof}

As concluded in~\cite{he2023learning}, this Lipschitz constant restriction is not particularly limiting in the classical setting since small values of~$\sigma$ can be used. 
However, in the quantum setting small values of~$\sigma$ introduce a large error if not enough shots are used. 
Specifically, consider the parameterised expectation 
(detailed in~\eqref{eq:basicparam})
\begin{equation}
    g(x) := 
\langle\Cc\rangle_{\Mg(x)\ket{\boldsymbol{0}}}
     = \braket{\boldsymbol{0}| \Mg(x)^{\dagger}\Cc\Mg(x)|\boldsymbol{0}} \in \mathbb{C},
    \qquad\text{for }x\in\RR^d.
\end{equation}
On actual quantum hardware we can only obtain an estimate~$\widetilde{g}(x)$ of~$g(x)$ using a finite number of shots, and we denote
\begin{equation}
    \varepsilon(x) := g(x) - \widetilde{g}(x)
\end{equation}
the (pointwise) inaccuracy, which is a random variable, since~$\widetilde{g}(x)$ is an empirical estimator.
This framework allows for error from both noisy circuits and estimating expectations using a finite number of shots.
Since the map~$g$ is bounded and the numbers of shots and qubits are finite, then there exists a constant $\varepsilon_f > 0$ such that $| \varepsilon(x)| \leq \varepsilon_f$ uniformly over $x \in \RR^d$.
Define $G(x) := \EE[g(x+ \delta)]$ 
and
$\widetilde{G}(x) := \EE
\left[\widetilde{g}(x+ \delta)\right]$, 
where $\delta_{\sigma} \sim \Nn(0,\Id\sigma^2)$, 
the following lemma provides a bound for the distance between the gradients of these two functions:

\begin{lemma}
The following bound holds for the Euclidean norm:
$$
\sup_{x\in\RR^d}\left\| \nabla G(x) - \nabla \widetilde{G}(x) \right\|
    \leq \frac{\sqrt{d}\, \varepsilon_f}{ \sigma}.
$$
\end{lemma}

\begin{proof}
From Theorem~\ref{thm:delg}, we can write, 
for any $x\in\RR^d$, 
\begin{align*}
\left\| \nabla G(x) - \nabla \widetilde{G}(x) \right\|
& = \left\| \frac{1}{\sigma^2}\EE[\delta_{\sigma} g(x+\delta_{\sigma})] - \frac{1}{\sigma^2}\EE\left[\delta_{\sigma}\widetilde{g}(x+\delta_{\sigma})\right]\right\| \\
& = \frac{1}{\sigma}\Big\|\EE[\delta_{1} g(x+\delta_{1}\sigma)] - \EE\left[\delta_{1}\widetilde{g}(x+ \delta_{1}\sigma)\right] \Big\| \\
& = \frac{1}{\sigma}\left\|\EE    [\delta_{1}\varepsilon(x+\delta_{1}\sigma)]  \right\|\\
& \leq \frac{1}{\sigma} \EE\|
    \delta_{1}\varepsilon(x+ \delta_{1}\sigma)\|
\leq \frac{\varepsilon_f}{\sigma}\EE\|\delta_{1}\|
    =\frac{\varepsilon_f \sqrt{2}}{\sigma}
    \frac{\Gamma(\frac{d+1}{2})}{\Gamma(\frac{d}{2})} 
    \leq\frac{\sqrt{d}\,\varepsilon_f}{\sigma}.
    \end{align*}
    We use Jensen's inequality to take the norm inside of the expectation. 
    The penultimate inequality follows since the Euclidean norm of the Gaussian is a Chi distribution, and the last inequality is Gautschi's~\cite[Equation (6)]{wendel1948note}.
\end{proof}
Therefore we see that decreasing  the parameter~$\sigma$ increases the impact of quantum induced sampling error.
Similar reasoning can be applied to finite difference based methods, and we refer the reader to ~\cite[Section~2.1]{berahas2021theoretical} for a general review of finite differences for gradients with error. 

\subsection{Random classical networks}
We call `random classical network'
a single-hidden-layer feedforward neural network with randomly generated internal weights,
where only the last layer of weights and hyperparameters is optimised over.
Such random networks have previously been successfully applied to solving different types of (high-dimensional) PDEs~\cite{gonon2023random, jacquier2023random, mattheakis2021unsupervised, raissi2019physics}.
For the Black-Scholes-type PDE (similar to the heat equation),
Gonon~\cite{gonon2023random}
provided a full error analysis of the prediction error of random neural networks.
Specifically, let $N \in \mathbb{N}$ be the number of hidden nodes, 
$\Bf = (B_1,\cdots,B_N)$ an iid sequence of real-valued random variables and
$\Ef=(E_1,\cdots,E_N)$ another iid sequence of random variables in~$\RR^{d}$, independent of~$\Bf$. 
For a vector~$\Wf \in \RR^N$, 
define the random function
$$
H_{\Wf}^{\Ef, \Bf}(\xx):=\sum_{i=1}^N W_i \varrho\left(E_i x + B_i\right),
\qquad\text{for all }x \in\RR^d,
$$
where we consider the non-linear activation function
$\varrho (x) = \max(0,x)$. 
The vector~$\Wf$ then represents the trainable parameters, 
while~$\Ef$ and~$\Bf$ are sampled from some prior distribution and frozen. 
These random networks are considerably easier to train than traditional fully connected deep neural networks,
especially in a supervised learning context, where training reduces to a linear regression.
The PINN algorithm is not a supervised learning approach, 
and this therefore does not apply;
however it does reduce the number of trainable parameters, 
and hence the computational burden.
In~\cite{gonon2021approximation}, 
Gonon, Grigoryeva and Ortega proved that, as long as the unknown function is sufficiently regular, it is possible to draw the internal weights of the random network from a generic distribution (not depending on the unknown object) and quantify the error in terms of the network architecture.

\subsection{Quantum neural networks}
\subsubsection{Quantum neural networks}
Using a quantum network for the PDE trial solution in the PINN algorithm was first proposed by Kyriienko, Paine and Elfving~\cite{Kyriienko_2021}. Essentially, the spatial variable~$x$ is embedded into a quantum state via a unitary operator~$\Uu(x)$ (referred to as the feature map), then a parameterised unitary operator $\Uu_\Theta$ (independent of~$x$) is applied before producing the output of the network by taking the expectation of a Hermitian cost operator~$\Cc$:
\begin{equation}
u_\Theta(x) := 
\langle\Cc\rangle_{\Uu_\Theta \, \Uu(x) \ket{\boldsymbol{0}}}
 = \braket{{\boldsymbol{0}|\left(  \Uu_\Theta\, \Uu(x) \right)^{\dagger}}\Cc \, \Uu_\Theta \, \Uu(x) |\boldsymbol{0}},
\end{equation}
and the parameters~$\Theta$ are found by minimising some loss function, for example~\eqref{eq:EmpLossFunction}.
In~\cite{Kyriienko_2021},
the authors suggested that increasing the number of qubits or using cost operators with more complex Pauli decompositions could produce more expressive networks.
Preliminary research~\cite{P_rez_Salinas_2021,PerezSalinas2020datareuploading,perez2021determining} has shown the potential of parameterising the feature maps and repeated application of unitary operators, leading to the more general formulation
\begin{equation}\label{eq:5}
u_\Theta(x) = \braket{{\boldsymbol{0}|\left( \prod_{i=L} ^{1} \Uu_{i,\theta_i} \Uu_{i,\omega_i}(x) \right)^{\dagger}}\Cc\left( \prod_{i=L}^{1} \Uu_{i,\theta_i} \Uu_{i,\omega_i}(x) \right)|\boldsymbol{0}},
\end{equation}
for some positive integer~$L$, where $\Theta = \{ \theta_1,\omega_1,\cdots,\theta_L,\omega_L\}$ is the set of all hyperparameters. 
In particular, the authors in~\cite{P_rez_Salinas_2021} showed that one-qubit networks can act as universal approximators for bounded complex continuous functions or integrable functions with a finite number of finite discontinuities. 

\subsubsection{Random quantum networks}\label{sec:RandomQNN}
Consider a system with~$n$ qubits. Let $\bm{A}: {\Omega} \to \mathbb{C}^{2^n \times 2^n}$ be a random function which maps the spatio-temporal variable $x$ to a unitary matrix and $\bm{\Lambda} \in \mathbb{C}^{2^n \times 2^n}$  a random unitary matrix distributed according to the Haar measure. Then for a suitable set of parameters~$\Theta$, we define the random function $u^{\bm{\Lambda},\bm{A}}_\Theta : \Omega \to \RR$,
\begin{equation}
u^{\bm{\Lambda},\bm{A}}_\Theta(x) 
:= \braket{{\boldsymbol{0}| \Bigl(  \Ug_{\Theta}(x)\bm{\Lambda}\bm{A}(x) \Bigl) ^{\dagger}}\Cc\Bigl(  \Ug_{\Theta}(x)\bm{\Lambda}\bm{A}(x)\Bigl)|\boldsymbol{0}},
\end{equation}
where~$\Cc$ is a Hermitian cost operator. 
When using this random quantum neural network to approximate~$u$ we generate~$\bm{A}$ and~$\bm{\Lambda}$, consider them fixed and train the parameters~$\Theta$.
Specific examples of~$\bm{A}$ are given in Section~\ref{sec:specificnetworks}. 
\begin{remark}
In practice, one may encode the data~$x$ only through~$\bm{A}$ and leave~$\Ug_{\Theta}$ independent of~$x$.
We leave the current formulation as is, allowing for more generality.
\end{remark}

\subsubsection{Optimised parameter shift}
When differentiating parameterised expectation values we apply the family of parameter shift rules discussed by Mari, Bromley and Killoran~\cite{mari2021estimating}.
In quantum computing, one-qubit rotation gates
can be written as $\exp\left\{-\frac{\I x}{2} \Gg\right\}$ for some unitary matrix~$\Gg$ and some real number~$x$.
We shall require here a slightly modified version:
\begin{definition}\label{def:rotationlike}
For $x\in\RR$, the matrix~$\Mg(x)$ is a single qubit rotation-like gate if
\begin{equation}
    \Mg(x) = \exp\left\{-\frac{\I x}{2} \Gg\right\},
\end{equation}
for some complex-valued involutory generator matrix~$\Gg$ (satisfying $\Gg^2 = \Id$).
\end{definition}
We shall use these single qubit rotation-like gates to construct an approximation~$u^{\bm{\Lambda},\bm{A}}_\Theta$.
The following lemma allows us to decompose the unitary conjugation $\Mg(x)^{\dagger}\, \Cc\, \Mg(x)$ to the sum of three easy-to-compute terms, 
as mentioned in~\cite[Equation~(5)]{mari2021estimating}, but without proof:
\begin{lemma}\label{lem:3Terms}
For any $x\in\RR$, the identity 
$\Mg(x)^{\dagger}\Cc\Mg(x) = \Ag + \Bg\cos(x) + \Cg\sin(x)$
holds for any single qubit rotation-like gate~$\Mg(x)$ with 
$$
\Ag = \frac{\Gg^\dagger\Cc\Gg + \Cc}{2},\qquad
\Bg = \frac{\Cc - \Gg^\dagger\Cc\Gg}{2},\qquad
\Cg = \frac{\I}{2}\left(\Gg^\dagger \Cc - \Cc\Gg\right).
$$
\end{lemma}

\begin{proof}
Let $x\in\RR$. 
Since~$\Mg(x)$ is a rotation-like gate as in Definition~\ref{def:rotationlike} with involutory generator~$\Gg$, then it is trivial to show that
$$
\Mg(x) = \exp\left\{-\frac{\I x}{2} \Gg\right\}
 = \cos\left(\frac{x}{2}\right)\Id - \I\sin\left(\frac{x}{2}\right)\Gg.
$$
Therefore
\begin{align*}
\Mg^{\dagger}\Cc\Mg
 & = \Big(\cos\left(\frac{x}{2}\right)\Id - \I\sin\left(\frac{x}{2}\right)\Gg\Big)^\dagger\Cc\Big(\cos\left(\frac{x}{2}\right)\Id - \I\sin\left(\frac{x}{2}\right)\Gg\Big)\\
 & = \Big(\cos\left(\frac{x}{2}\right)\Id + \I\sin\left(\frac{x}{2}\right)\Gg^\dagger\Big)\Cc\Big(\cos\left(\frac{x}{2}\right)\Id - \I\sin\left(\frac{x}{2}\right)\Gg\Big)\\
 & = \left[\cos\left(\frac{x}{2}\right)^2\Cc + 
 \sin\left(\frac{x}{2}\right)^2\Gg^\dagger\Cc\Gg\right]
 + \I\left[\sin\left(\frac{x}{2}\right)\cos\left(\frac{x}{2}\right)\Gg^\dagger\Cc
  - \sin\left(\frac{x}{2}\right)\cos\left(\frac{x}{2}\right)\Cc\Gg\right]\\
 & = \left[\cos\left(\frac{x}{2}\right)^2\Cc + 
 \left(1-\cos\left(\frac{x}{2}\right)^2\right)\Gg^\dagger\Cc\Gg\right]
 + \frac{\I}{2}\left(\Gg^\dagger \Cc - \Cc\Gg\right) \sin(x)\\
& = \Gg^\dagger\Cc\Gg + \left(\Cc - \Gg^\dagger\Cc\Gg\right)\cos\left(\frac{x}{2}\right)^2  + \frac{\I}{2}\left(\Gg^\dagger \Cc - \Cc\Gg\right) \sin(x)\\
& = \Gg^\dagger\Cc\Gg + \left(\Cc - \Gg^\dagger\Cc\Gg\right)\frac{1+\cos(x)}{2}  + \frac{\I}{2}\left(\Gg^\dagger \Cc - \Cc\Gg\right) \sin(x)\\
& = \frac{\Gg^\dagger\Cc\Gg + \Cc}{2}
+\frac{\Cc - \Gg^\dagger\Cc\Gg}{2}\cos(x)
+ \frac{\I}{2}\left(\Gg^\dagger \Cc - \Cc\Gg\right) \sin(x),
\end{align*}
and the lemma follows.
\end{proof}
Consider the function $g:\RR\to\RR$ defined as
\begin{equation}\label{eq:basicparam}
    g(x) := \braket{{\boldsymbol{0}| \Mg(x)^{\dagger}\, \Cc\, \Mg(x)|\boldsymbol{0}}},
    \qquad\text{for any }x\in\RR,
\end{equation}
given some single qubit rotation-like matrix~$\Mg$ as in Definition~\ref{def:rotationlike}.
Clearly~$g$ is infinitely differentiable, and 
it is furthermore periodic by Lemma~\ref{lem:3Terms}.
Denote its partial derivatives 
\begin{equation}
    g_{j_1,j_2,\cdots,j_d} (x)
    := \frac{\partial ^d g(x)}{\partial x_{j_1} \partial x_{j_2}\cdot \cdot \cdot \partial x_{j_d}},
    \qquad\text{for any }x\in\RR.
\end{equation}
Applying the family of parameter shift rules~\cite[Equation~(9)]{mari2021estimating} results in
\begin{equation}\label{eq:paramshift}
    g_{j} (x) = \frac{g(x+s\ef_j)-g(x-s\ef_j)}{2\sin(s)},
    \qquad\text{for any }x\in\RR,
\end{equation}
for any $s \in \RR\setminus\{k\pi:k\in\mathbb{Z}\}$, 
where~$\ef_j$ is a unit vector along the $x_j$ axis.
Iteratively applying this rule with the same shift results in
\begin{equation}
    g_{j_1,j_2}(x) = \frac{g(x+s(\ef_{j_1}+\ef_{j_2}))- g(x+s(-\ef_{j_1}+\ef_{j_2}))
    -g(x+s(\ef_{j_1} - \ef_{j_2}))+g(x-s(\ef_{j_1}+\ef_{j_2}))}{4\sin(s)^2},
\end{equation}
for any $x\in\RR$.
For $j_1 = j_2$ and using the value $s=\frac{\pi}{2}$ reduces this to
\begin{equation}\label{eq:reduced}
    g_{j,j} (x) = \frac{g(x+\pi \ef_j)-g(x)}{2},
\end{equation}
and for $s = \frac{\pi}{4}$ we obtain
\begin{equation}\label{eq:2}
g_{j,j} (x)
= \frac{g\left(x+\ef_j\frac{\pi}{2}\right)-2g(x) + g\left(x-\ef_j\frac{\pi}{2}\right)}{2}.
\end{equation}
While~\eqref{eq:reduced} only requires the evaluation of two expectation values compared to the three of~\eqref{eq:2}, the latter has the distinct advantage of providing both the derivatives~$g_{j,j}$ and~$g_{j}$.
For the complexity analysis we apply either~\eqref{eq:reduced} or~\eqref{eq:2} depending on which one is more efficient for the chosen PDE. 
The optimised parameter shift rules above are covered in~\cite{mari2021estimating}, but simple yet tedious computations can provide higher other derivatives if needed, as shown in the following:
\begin{theorem}\label{thm:reduced5}
Let~$g$ be any function of the form~\eqref{eq:basicparam} with~$\Mg$ a single qubit rotation-like gate and $\Cc$ a Hermitian cost operator.
For any $d \geq 2$ and all $x \in \RR$, 
\begin{equation}\label{eq:reduced5}
\frac{\partial^d g(x)}{\partial x_j ^d}  = \left[2\sin\left(\frac{\pi}{d}\right)\right]^{-d}
\left\{\left(1 + (-1)^d\right)
g(x + \pi \ef_j) + 
\sum_{i=1}^{d-1}  (-1)^{i+d}{d \choose {i}}g\left(x + 
\left(\frac{2\pi i}{d}-\pi\right)\ef_j\right)
\right\}.
\end{equation}
\end{theorem}

\begin{proof}
Repeated applications of the parameter shift rule~\eqref{eq:paramshift} with the same basis vector and the same shift magnitude for each shift results in    \begin{equation}\label{eq:nonreduced}
       \frac{\partial ^d}{\partial x_j ^d} g(x) = \sum_{i=0}^{d}  (-1)^{i+d}{d \choose {i}}\frac{g\left(x + 
        (\frac{2\pi i}{d}-\pi)\ef_j\right)}{(2\sin(\pi/d))^d}.
    \end{equation}
The first and last terms in the expansion are
    \begin{equation}
        \frac{1}{(2\sin(\pi/d))^d}\left(
        (-1)^d g(x-\pi \ef_j) + g(x+\pi \ef_j)\right).
    \end{equation}
By periodicity of our expectation function, 
for~$d$ even this reduces to
$\frac{2g(x+\pi \ef_j)}{(2\sin(\pi/d))^d}$,
whereas when~$d$ is odd the first and last terms in the expansion cancel.
\end{proof}
We see that we have the ability to find a $d^{\text{th}}$ order unmixed partial derivative with just $d$ evaluations of the parameterised expectation. Quantum automatic differentiation engines can be built using a combination of the chain rule and the above formulae where the argument of each expectation is calculated on a classical computer, the expectation is calculated on a quantum computer before the result is fed back to the classical computer to compile the different values in the correct way to construct the derivative. 
This suggested method is in sharp contrast to the automatic differentiation procedures carried out on classical neural networks where different orders of derivatives can only be calculated sequentially. For traditional automatic differentiation the algorithm has to first build the computational graph for first-order derivatives
and then perform back-propagation on the graph for the second-order derivatives before this process is repeated.
Note that previous publications such as~\cite{Kyriienko_2021} suggested repeated applications of the basic parameter shift rule which does not benefit from the computational acceleration discussed above.

\section{Random network architecture and complexity analysis}\label{sec:RNN}
\subsection{Random network architecture}\label{sec:specificnetworks}
We assume the input data $x$ is scaled to the interval $[0,1]$. 
By repeated applications of the so-called UAT (Universal Approximation) gate, 
one-qubit circuits have the ability to approximate any bounded complex function on the unit hypercube~\cite{P_rez_Salinas_2021}. 
This UAT gate is defined as
\begin{equation}
    \Ug^{\text{UAT}}_\Theta(x) := \Rg_{y}(2\varphi)\Rg_{z}(2\gamma^\top x +2\alpha),
    \qquad\text{for all } x \in [0,1]^d,
\end{equation}
with $\Theta = (\varphi,\gamma,\alpha) \in \RR\times\RR^d\times\RR$. 
This was extended in~\cite{PerezSalinas2020datareuploading} to multiple qubits by applying the UAT gate to each qubit followed by an entangling layer. 
We create the random network's trainable layer using this UAT gate as well as an entangling layer, which can be repeated several times. 
For the entangling layer we choose the sequence
\begin{equation}
\prod_{i= n -1}^{0}\Cnotg(q_i,q_{i+1 \, \mathrm{mod} \, n}),
\end{equation}
of controlled \texttt{NOT} gates, 
where $\Cnotg(q_i,q_j)$ denotes the controlled \texttt{NOT} gate acting on qubit~$q_j$ with control qubit~$q_i$. 
The whole trainable $M$-layer circuit then reads
\begin{equation}\label{fullcircuit}
\Ug_{\Theta}(x)
:= \prod_{j= 1}^{M}\left(\prod_{i= n -1}^{0}\Cnotg(q_i,q_{i+1 \, \mathrm{mod} \, n})\right)\left( 
\bigotimes_{i= 0}^{n -1} \Ug^{\text{UAT}}_{\Theta_{i,j}}(x)\right),
\end{equation}
and we write $\Theta = \{\Theta_{i,j}\}_{0 \leq i \leq n-1, 1 \leq j \leq M }$.
For the `encoding' matrix~$\bm{A}(x)$ introduced in Section~\ref{sec:RandomQNN}, we use
\begin{equation}\label{eq:rot}
\bm{A}(x) := \bigotimes_{j = 1 }^{n}
\Rg_{z}\left[\mathcal{X}_j\pi\acos\left(\frac{3}{2} \left(x_{\mathfrak{g}_{j}}-\half\right)\right)\right], 
\end{equation}
where $(\mathcal{X}_j)_{j=1,\ldots,n}$ are iid $\mathcal{U}(0,1)$ and an index $\mathfrak{g}_{j} := 1 + (j+1 \mod d)$ with nonlinear activation function~$\acos(\cdot)$.
The constants inside~$\acos(\cdot)$ and the constant~$\pi$ inside~$\Rg_{z}$ are justified experimentally as arguments of~$\Rg_{z}$ away from $\pm \pi$ lead to better results.
Such a choice of encoding matrix is inspired from~\cite{Kyriienko_2021},
who show that the resulting matrix~$\bm{A}(x)$ can then be written as a Chebyshev polynomial, the order of which increases with~$n$, known for approximating well non-linear functions.

Finally, we use a specific Ising Hamiltonian with transverse and longitudinal magnetic fields and homogeneous Ising couplings for the cost operator:
\begin{equation}\label{eq:IsingHamilt}
    \Cc = \sum_j \left[ \widehat{\Zg}_{j \Mod{n}
}\widehat{\Zg}_{j+1 \Mod{n}}
+ \widehat{\Zg}_{j \Mod{n}} + \widehat{\Xg}_{j\Mod{n}} \right],
\end{equation}
where~$\Xg$ and~$\Zg$ are the usual one-qubit Pauli gates and the index refers to the qubit index they act upon.
This is a relatively complex cost operator, allowing us to approximate a wide class of functions, better than  the simpler cost operator $\sum_j \widehat{\Zg}_j$.
Therefore the output of the quantum network is 
\begin{equation}\label{fullquantnetwork}
u_\Theta(x) :=  
\langle\Cc\rangle_{\Ug_{\Theta}(x) \ket{\boldsymbol{0}}}
 = \braket{{\boldsymbol{0}|\left(  \Ug_{\Theta}(x) \right)^{\dagger}}\Cc \, \Ug_{\Theta}(x) |\boldsymbol{0}}.
\end{equation}

\subsection{Complexity analysis}\label{sec:Complexity}
Define the auxiliary variables~$\alpha_i$ and~$\alpha_{i,j}$ 
as the number of single-qubit rotation-like gates 
respectively with~$x_i$ dependence and with~$x_i$ and~$x_j$ dependence in the circuit for~$u_\Theta$. 
These are not uniquely defined values since the decomposition into single-qubit rotation-like gates is not unique.
Given the loss function~\eqref{eq:EmpLossFunction}, 
define the quantity $\xi(\mathcal{L}_{n_e,n_r})$, 
which returns the total number of~$u_\Theta$ evaluations needed to calculate the loss function $\mathcal{L}_{n_e,n_r}$,
and which we will use as our complexity metric.
Let $N_{\Theta}$ be the total number of components in~$\Theta$. It may seem pertinent to include in this metric the cost of computing all the derivatives $(\partial_{\Theta_i}\mathcal{L})_{i=1,\ldots,N_{\Theta}}$
(for example for the optimisation part);
this is however unnecessary since the exact number of~$u_\Theta$ evaluations to do so is given by
$2N_{\Theta}\xi(\mathcal{L})$.

\subsection{The p-Laplace equation}\label{variational1}
Consider the $p$-Laplace ($1 < p < \infty$) 
equation with Dirichlet boundary conditions:
\begin{align}\label{eq:poiss}
 \begin{cases}
 \Delta_p u + f = 0, & \text{on }\Omega, \\ 
 u=h, & \text{on }\partial \Omega,
 \end{cases}
\end{align}
where $u \in \Ww_0^{1, p}(\Omega) \cap L^p(\Omega)$, $f \in L^{\frac{p}{p-1}}(\Omega)$ and $h$ the trace of some $\Ww_0^{1, p}(\Omega) \cap L^p(\Omega)$ function. 
Where the $p$-Laplace operator reads
\begin{equation}\label{eq:LaplaceOperator}
    \Delta_p u := |\nabla u|^{p-4}\left(|\nabla u|^{2}\Delta u + (p-2)\sum_{i,j =1}^{d} \frac{\partial u}{\partial x_i}\frac{\partial u}{\partial x_j}\frac{\partial^2 u}{\partial x_i \partial x_j}\right).
\end{equation}
For the variational formulation we have~\cite[Lemma~(2.3)]{bhuvaneswari2012weak}
\footnote{This is for the case with the 
zero boundary condition.
However any problem with prescribed
nonzero boundary values can easily be transformed into this setting~\cite[Section~(6.1.2)]{evans2022partial}}

\begin{equation}\label{weak11}
    u \in \argmin_{%
      \substack{%
        v \in \Ww^{1,p}(\Omega) \\
        \Tr v = h
      }} \left (\frac{1}{p}\int_\Omega |\nabla v|^p \D x - \int_\Omega f v \D x\right).     
\end{equation}

We refer the reader to~\cite{evans2022partial} for a reference on the weak formulation of PDEs.
To translate the variational statement~\eqref{weak11} to a loss function,
we add a penalisation term for the boundary conditions and approximate the integral along the sample points, resulting in the empirical loss function
\begin{equation}\label{eq:EmpLossFunctionWeak}
\mathcal{L}^{\text{Var}}_{n_e,n_r}(u_\Theta)
 := \frac{\lambda_e}{n_e}\sum_{i=1}^{n_e}\left|
    u_\Theta\left(X^{(e)}_i\right) - h\left(X^{(e)}_i\right)   \right|
    + \frac{1}{n_r}\left (\frac{1}{p} |\nabla u_\Theta (X^{(r)}_i)|^p -  f(X^{(r)}_i) u_\Theta(X^{(r)}_i) \right).
\end{equation}

This idea of minimising a functional to solve PDEs using neural networks has previously been studied for Poisson equations~\cite{e2018deep,muller2020deep}. 
The cases $p\ne 2$ and $p=2$ need to be studied separately 
since the mixed second-order partial derivatives of the $p$-Laplace operator~\eqref{eq:LaplaceOperator} cancel when $p=2$.

\subsubsection{The $p \ne 2$ case}

\begin{lemma}\label{proofcomplex1}
For the prototypical PINN loss function~\eqref{eq:EmpLossFunction}, we have
$$
\xi(\mathcal{L}_{n_e,n_r}) = n_r
    \sum_{i,j = 1, i \leq j}^{d} 
    \left(4\alpha_i \alpha_j - \alpha_{i,j} \right)   + n_e,
$$
provided the second-order partial derivatives all commute, namely when
\begin{equation}
    \left[\frac{\partial v_{\Theta}}{\partial x_i},\frac{\partial v_{\Theta}}{\partial x_j} \right] = 0, \quad \text{for all } i,j \leq d.
\end{equation}
Without any assumptions on the partial derivatives, we have
$$
\xi(\mathcal{L}_{n_e,n_r})
= n_r
\sum_{i,j = 1}^d \left(4\alpha_i \alpha_j - \alpha_{i,j} \right) + n_e.
$$
\end{lemma}
\begin{proof}
Consider decomposing the quantum circuit~\eqref{fullcircuit} responsible for producing $u_\Theta(x)$ into one with just single qubit rotation-like gates and CNOTs. Assuming~$p$ single qubit rotation gates and with a slight abuse of notation,
\begin{equation}\label{decomp}
    u_\Theta (x) = g(\phi_1(x),\cdots,\phi_{p} (x)) = g(\phi(x)),
\end{equation}
where each $\phi_i$ is the rotation angle for a particular single qubit rotation-like gate. 
Note that this decomposition is clearly not unique, as one could choose a decomposition based on which has the lowest~$\xi$ value.
Basic applications of the chain rule then yield

\begin{subequations}
  \begin{align}    
          \frac{\partial u_\Theta }{\partial x_i} &= \sum_a \frac{\partial g}{\partial \phi_a } \frac{\partial \phi_a}{\partial x_i}\label{eq-a},\\
       \frac{\partial u_\Theta }{\partial x_i \partial x_j} &= \sum_{a,b} \frac{\partial g}{\partial \phi_b \partial \phi_a}\frac{\partial \phi_a}{\partial x_i}\frac{\partial \phi_b}{\partial x_j} + \sum_a \frac{\partial g}{\partial \phi_a} \frac{\partial \phi_a}{\partial x_i \partial x_j}.\label{eq-b}
  \end{align}
\end{subequations}

We then split the sum over $a,b$ up resulting in
$$
\frac{\partial u_\Theta }{\partial x_i \partial x_j} = \sum_{a \neq b} \frac{\partial g}{\partial \phi_b \partial \phi_a}\frac{\partial \phi_a}{\partial x_i}\frac{\partial \phi_b}{\partial x_j} + \sum_a \frac{\partial g}{\partial \phi_a\partial \phi_a}\frac{\partial \phi_a}{\partial x_i}\frac{\partial \phi_a}{\partial x_j} + \sum_a \frac{\partial g}{\partial \phi_a} \frac{\partial \phi_a}{\partial x_i \partial x_j}.
$$
In the first summation there are 
$(\alpha_i \alpha_j - \alpha_{i,j})$ terms, 
so applying the standard parameter shift rule equation~\eqref{eq:paramshift} twice results in $4(\alpha_i \alpha_j - \alpha_{i,j})$ evaluations of~$u_\Theta$.
In the second summation there are~$\alpha_{a,b}$ terms, so the optimised parameter shift rule~\eqref{eq:2} results in $3\alpha_{i,j}$ evaluations of~$u_\Theta$.
We use this parameter shift rule as it provides all of the quantum gate partial derivatives needed for the third summation too.

For the $p$-Laplace operator~\eqref{eq:LaplaceOperator}, we require all mixed partial derivatives, that is the derivatives for all pairs $(i,j)$; as a result the number of operators needed to evaluate all these derivatives of~$u_\Theta$ is equal to
\begin{equation}
   \sum_{i,j=1}^d \left( 4(\alpha_i \alpha_j - \alpha_{i,j}) + 3\alpha_{i,j} \right) =    \sum_{i,j=1}^d \left( 4\alpha_i \alpha_j - \alpha_{i,j}  \right),
\end{equation}
where we have assumed the derivatives do not necessarily commute.
Assuming they do, we then only need to sum over pairs with $i \leq j$, as
\begin{equation}
 \sum_{i,j = 1, i \leq j}^d
 \left( 4\alpha_i \alpha_j - \alpha_{i,j}  \right).
\end{equation}
Note that the $n_r$ factor comes from the number of times the residual must be evaluated;
since boundary data appears with no derivatives, 
each sample then only requires one~$u_\Theta$ evaluation.
\end{proof}

\begin{lemma}
For the variational loss function formulation~\eqref{eq:EmpLossFunctionWeak},
\begin{equation}\label{weakplaplace1}
\xi\left(\mathcal{L}^{\text{Var}}_{n_e,n_r}\right) = n_r\left( 1 + 2\sum_{i=1}^{d} \alpha_i\right) + n_e.
\end{equation}
\end{lemma}
\begin{proof}
The minimisation statement that arises from the variational formulation~\eqref{weak11} involves 
both the gradient and the function value for each sample point.
Using the same decomposition as before~\eqref{decomp}, the chain rule~\eqref{eq-a} and the basic parameter shift rule~\eqref{eq:paramshift}
the gradient involves exactly
$2\sum_{i=1}^{d} \alpha_i$
evaluations of~$u_\Theta$.
The function value results in only one~$u_\Theta$ evaluation.
This is done for each sample point inside the domain resulting in
\begin{equation}
\xi\left(\mathcal{L}^{\text{Var}}_{n_e,n_r}\right) = n_r\left( 1 + 2\sum_{i=1}^{d} \alpha_i\right).
\end{equation}
The extra~$n_e$ term in~\eqref{weakplaplace1} is a result of the boundary penalisation term in~\eqref{eq:EmpLossFunctionWeak}.
\end{proof}
\begin{lemma}\label{lem:GaussSmoothError}
If the PDE trial solution is the Gaussian smoothed output of a quantum network then, using~$K$ classical samples of Gaussian noise,
\begin{align}\label{Gausscomplex1}
    \xi\left(\mathcal{L}_{n_e,n_r}^{\text{Smooth}}\right) = K(5n_r+n_e).
\end{align}
\end{lemma}

\begin{proof}
When the output of the quantum network is Gaussian smoothed~\eqref{smoothing1}, the Hessian matrix $\bm H f_\Theta$ can be calculated using
\begin{equation}\label{baselaplacian}
    \bm{H} f_\Theta(x) =\EE\left[\left(\frac{\delta \delta^{\top}-\sigma^2 \Id}{\sigma^4}\right) u_\Theta(x+\delta)\right],
    \qquad\text{for all }x \in \Omega,
\end{equation}
with $\delta \sim \Nn\left(0, \sigma^2 \Id\right)$, as proven in~\cite{he2023learning}. The general variance of this estimator can be reduced by applying the control variate and antithetic variable method. For the additive control variate we use $u_\Theta(x)$ as a baseline which turns the estimate~\eqref{baselaplacian} into
\begin{equation}\label{baselaplacian2}
    \bm{H} f_\Theta(x) =\EE\left[\left(\frac{\delta \delta^{\top}-\sigma^2 \Id}{\sigma^4}\right)\Big\{u_\Theta(x+\delta)-u_\Theta(x)\Big\}\right],
    \qquad\text{for all }x \in \Omega.
\end{equation}
just as authors do in~\cite{he2023learning}. Notice the estimate in~\eqref{baselaplacian} and~\eqref{baselaplacian2} is invariant when $\delta$ is substituted with $-\delta$, averaging this new estimator with the one in~\eqref{baselaplacian2} results in
\begin{equation}
        \bm{H} f_\Theta(x) =\EE\left[\left(\frac{\delta \delta^{\top}-\sigma^2 \Id}{2\sigma^4}\right)
        \Big\{u_\Theta(x+\delta)+ u_\Theta(x-\delta) - 2u_\Theta(x)\Big\}\right].
\end{equation}

The same technique can be applied to the estimator for the gradient resulting in
\begin{equation}
    \nabla_x f_\Theta(x) = \EE\left[\frac{\delta}{2 \sigma^2}(u_\Theta(x+\delta)-u_\Theta(x-\delta))\right].
\end{equation}
There are a total of five $u_\Theta$ terms in both expressions, the boundary term will feature a single $f_\Theta$ term for each sample point. Assuming we approximate each classical expectation with $K$ iid Gaussian samples the expression~\eqref{Gausscomplex1} follows. 
\end{proof}

\begin{lemma}
We can combine a Gaussian smoothed trial solution with the variational loss function reformulation \eqref{eq:EmpLossFunctionWeak} and \eqref{weak11} resulting in
\begin{align}
    \xi\left(\mathcal{L}_{n_e,n_r}^{\text{Var,Smooth}}\right) = K(3n_r+n_e).
\end{align}
\end{lemma}

\begin{proof}\label{proofgausweak}
For each sample point the variational formulation~\eqref{weak11} only involves a gradient term and the function value 
and, proceeding as in the proof of Lemma~\ref{lem:GaussSmoothError},
this produces~$3K$ evaluations of~$u_\Theta$, where we again assume~$K$ iid Gaussian samples. 
Once more the boundary penalisation element shown in~\eqref{eq:EmpLossFunctionWeak} produces the~$Kn_e$ term.
\end{proof}

\subsubsection{The $p=2$ case}
The $p=2$ Laplace equation reduces to the Poisson equation, and the following holds:

\begin{lemma}\label{proofcomplexlaplace}
For the Poisson equation,
the complexity metrics for the prototypical loss formulation~\eqref{eq:EmpLossFunction} and the loss function using the variational form~\eqref{eq:EmpLossFunctionWeak} read
\begin{align}
\xi(\mathcal{L}_{n_e,n_r}) &= n_r\left(\sum_{i=1}^d \left( 4\alpha^2_i - \alpha_i \right)\right) + n_e, \\ \xi\left(\mathcal{L}^{\text{Var}}_{n_e,n_r}\right) &= n_r\left( 1 + 2\sum_{i=1}^{d} \alpha_i\right) + n_e.
\end{align}
\end{lemma}

\begin{proof}
The proof of $\xi\left(\mathcal{L}^{\text{Var}}_{n_e,n_r}\right)$ is the exact same as the case with $p \neq 2$. The first statement's proof is very similar to that in Lemma~\ref{proofcomplex1} however we include it for the sake of completeness. Using the same decomposition as in Lemma~\ref{proofcomplex1} we have

$$\frac{\partial u_\Theta }{\partial x_i \partial x_j} = \sum_{a \neq b} \frac{\partial g}{\partial \phi_b \partial \phi_a}\frac{\partial \phi_a}{\partial x_i}\frac{\partial \phi_b}{\partial x_i} + \sum_a \frac{\partial g}{\partial \phi_a \partial \phi_a}\frac{\partial \phi_a}{\partial x_i}\frac{\partial \phi_a}{\partial x_i} + \sum_a \frac{\partial g}{\partial \phi_a}\frac{\partial}{\partial x_i} \frac{\partial \phi_a}{\partial x_i}.
$$

For the first term we apply the basic parameter shift rule twice to the $\alpha^2_i - \alpha_i$ terms resulting in $4(\alpha^2_i - \alpha_i)$ evaluations of $u_\Theta$ . For the second term we apply the optimised parameter shift rule given in~\eqref{eq:2}, producing $3\alpha_i$ evaluations of $u_\Theta$ and also providing all of the quantum derivatives needed for the third term. Summing over all~$i$, 
\begin{equation}
    \sum_{i=1}^d \left( 4(\alpha^2_i - \alpha_i) + 3\alpha_i\right) =\sum_{i=1}^d \left( 4\alpha^2_i - \alpha_i \right).
\end{equation}
We then add the boundary term and multiply by the relevant sample sizes to produce the original statement. 
\end{proof}
\begin{lemma}
If the PDE trial solution is Gaussian smoothed, then
\begin{align}\label{poissgaus}
    \xi\left(\mathcal{L}_{n_e,n_r}^{\text{Smooth}}\right) = K(3n_r + n_e).
\end{align}
\end{lemma}
\begin{proof}
    The Gaussian Laplacian term is given by~\cite[Equation (11)]{he2023learning}
\begin{equation}\label{gausslaplace1}
    \Delta f_\Theta(x) = \EE\left[\left(\frac{\|\delta\|^2-\sigma^2 d}{2 \sigma^4}\right)(u_\Theta(x+\delta)+u_\Theta(x-\delta)-2 u_\Theta(x))\right].
\end{equation}
There are a total of three~$u_\Theta$ terms in~\eqref{gausslaplace1}, the boundary term will feature a single $f_\Theta$ term for each sample point. Assuming we approximate each classical expectation with $K$ iid Gaussian samples the expression~\eqref{poissgaus} follows. 
\end{proof}

\subsection{Heat equation}
We now consider the heat equation, solution to the system
\begin{align}\label{eq:heat}
 \begin{cases}
 \left(\Delta - \dt\right) u(t,x) = 0, & \text{for all }x \in \Omega , t \in [0,T),\\ 
 u(0,x) = h(x), & \text{for all } x \in \Omega,\\
 u(t,x) = g(x), & \text{for all }x \in \partial\Omega,  t \in [0,T).
 \end{cases}
\end{align}
Similarly to above, we define $\alpha_t$ to be the number single qubit rotation-like gates with time dependence.

\begin{lemma}
The complexity with the loss function~\eqref{eq:EmpLossFunction} reads
\begin{align}
    \xi(\mathcal{L}_{n_e,n_r}) &= n_r\left(2\alpha_t+\sum_{i=1}^d \left( 4\alpha^2_i - \alpha_i \right)\right) + n_e.
\end{align}
If the network is Gaussian smoothed then
\begin{align}\label{Gausscomplex5}
    \xi\left(\mathcal{L}_{n_e,n_r}^{\text{Smooth}}\right) = K(5n_r+n_e),
\end{align}
\end{lemma}
\begin{proof}
    For the first statement we apply the working from proof of the $p$-Laplace with $p=2$ followed by the basic parameter shift rule for $\dt u_\Theta$ which produces the $2\alpha_t$ factor. 
    For the Gaussian statement we use the Laplacian term in~\eqref{gausslaplace1} and the corresponding variance reduced equation reads
    \begin{equation}\label{Complexfullfirst}
        \dt f_\Theta=\EE_{(\delta_t,\delta_x) \sim \Nn(0,\Id\sigma^{2})}
        \left[\frac{\delta_t}{2 \sigma^2}\Big(u_\Theta(t + \delta_t,x+\delta_x)-u_\Theta(t - \delta_t,x-\delta_x)\Big)\right]
    \end{equation}
    for the time derivative. Counting terms we arrive at~\eqref{Gausscomplex5}.
\end{proof}
Since the other components of the Gaussian noise in~\eqref{Complexfullfirst} all return first-order derivatives, then
$$
(\dt f_\Theta,\nabla_x f_\Theta)=\EE_{(\delta_t,\delta_x) \sim \Nn(0,\Id\sigma^{2})}\left[\frac{(\delta_t,\delta_x)}{2 \sigma^2}(u_\Theta(t + \delta_t,x+\delta_x)-u_\Theta(t - \delta_t,x-\delta_x))\right].
$$

\subsection{Hamilton–Jacobi equation}
Consider the classical linear-quadratic Gaussian (LQG) control problem in~$d$ dimensions, with the associated HJB equation given by~\cite{Han_2018}
\begin{align}\label{eq:HJB}
 \begin{cases}
 \left(\dt u + \Delta u - \mu \|\nabla_x u\|^2\right)(t,x)
 = 0, & \text{for all } x \in \RR^{d} , t \in [0,T),\\
 u(T,x) = h(x), & \text{for all }x\in \RR^{d}.
 \end{cases}
\end{align}
with $\mu \in \RR$. 
Similarly to the $p$-Laplace case above, we have the following:

\begin{lemma}
In this HJB framework, the complexity with the standard loss function~\eqref{eq:EmpLossFunction} reads
\begin{align}\label{complexHJB}
    \xi(\mathcal{L}_{n_e,n_r}) &= n_r\left(2\alpha_t+ 3\sum_{i=1}^{d} \alpha_i(2\alpha_i -1)\right) + n_e,
\end{align}
which simplifies, for a Gaussian smoothed network,
to
\begin{align}\label{GaussHJHB}
\xi\left(\mathcal{L}_{n_e,n_r}^{\text{Smooth}}\right) = K(5n_r+n_e),
\end{align}
\end{lemma}
\begin{proof}
We apply the same decomposition and parameter shift applications as is done in Lemma~\ref{proofcomplexlaplace} and note that doing so provides all of the quantum expectations needed to calculate $\nabla_x u_\Theta$ too. The basic parameter shift rule is applied to find $\dt u_\Theta$, after multiplying each term by the number of relevant samples and adding the boundary element we recover~\eqref{complexHJB}.
For the Gaussian expression we apply~\eqref{complexHJB} and~\eqref{gausslaplace1}, counting terms adding the boundary term and multiplying by the relevant sample sizes produces~\eqref{GaussHJHB}.
\end{proof}

\subsection{Explicit comparison}
With the random network introduced in Section~\ref{sec:specificnetworks}, assuming that the number~$n$ of qubits is a multiple of the problem dimension~$d$,
we then have, for each $i, j \in 1,\cdots,d$,
$$
\alpha_i = \frac{n}{d} + Mn 
    \qquad\text{and}\qquad
\alpha_{i,j} = Mn.
$$
For this particular class of networks, all second-order partial derivatives commute. 
For the $p$-Laplace equation with $p \neq 2$ and without Gaussian smoothing the 
complexity in Lemma~\ref{proofcomplex1} reads
\begin{align*}
\xi(\mathcal{L}_{n_e,n_r}) &= n_r
\left[\sum_{i,j = 1, i \leq j}^d
4\left[\frac{n}{d} + Mn\right]^2 - Mn \right] 
+ n_e
  = \frac{n_r}{2}d(d+1)\left[4\left[\frac{n}{d} + Mn\right]^2 - Mn\right] + n_e,\\
\xi\left(\mathcal{L}^{\text{Var}}_{n_e,n_r}\right) &= n_r\left( 1 + 2\sum_{i=1}^{d} \alpha_i\right) + n_e = n_r\left( 1 + d(d+1)\left(\frac{n}{d} + Mn\right)\right) + n_e.
\end{align*}
Let $\mathcal{L}^{\text{Var, Smooth}}_{n_e,n_r}$ represent the loss function~\eqref{eq:EmpLossFunctionWeak}
using Gaussian smoothing~\eqref{smoothing1} 
in its variational form~\eqref{weak11}. 
In this case, we have
$$
\xi(\mathcal{L}_{n_e,n_r}^{\text{Smooth}}) = K(5n_r+n_e)
    \qquad\text{and}\qquad
\xi\left(\mathcal{L}^{\text{Var,Smooth}}_{n_e,n_r}\right) = K(3n_r+n_e).
$$
Fix two final UAT trainable layers ($M=2$) and  let the number of qubits be three times the problem's dimension ($n=3d$). 
The number of evaluations then reduces to
\begin{equation*}
\begin{array}{rcll}
\displaystyle \xi\left(\mathcal{L}_{n_e,n_r}\right)
& = &
\displaystyle n_rd(d+1)(2(3 + 6d)^2 - 3d) + n_e 
& = \displaystyle n_r(72d^4 + 141d^3 + 87d^2 + 18d )+n_e,\\
\\
\displaystyle \xi\left(\mathcal{L}^{\text{Var}}_{n_e,n_r}\right)
& = & 
\displaystyle n_r\left( 1 + d(d+1)(3 + 6d)\right) + n_e
 & = \displaystyle n_r(6d^3 + 9d^2 + 3d + 1) + n_e.
\end{array}
\end{equation*}
While still polynomial we clearly have better scaling in the dimension for the variational formulation. 
We in particular have
$\xi(\mathcal{L}_{n_e,n_r}) > \xi\left(\mathcal{L}^{\text{Var}}_{n_e,n_r}\right)$
for all values of $d \in \mathbb{N}$.
In~\cite{he2023learning}, the authors find good experimental results using values of $K\in[256,2048]$ for classical PINNs,
so we now fix $K = 1024$.
For $n_r = n_e$, we have
\begin{equation*}
\begin{array}{rll}
\displaystyle \xi\left(\mathcal{L}^{\text{Var, Smooth}}_{n_e,n_r}\right) 
& \displaystyle <\xi\left(\mathcal{L}_{n_e,n_r}^{\text{ Smooth}}\right), & \text{for all }d\in\mathbb{N},\\
\displaystyle \xi\left(\mathcal{L}^{\text{Var, Smooth}}_{n_e,n_r}\right)
& \displaystyle < \xi\left(\mathcal{L}^{\text{Var}}_{n_e,n_r}\right), & \text{if and only if }d \geq 9,\\
\displaystyle \xi\left(\mathcal{L}_{n_e,n_r}^{\text{ Smooth}}\right)
& \displaystyle     < \xi\left(\mathcal{L}^{\text{Var}}_{n_e,n_r}\right), & \text{if and only if }d \geq 10.
    \end{array}
\end{equation*}
For the $p$-Laplace equation,
Gaussian smoothing is thus more efficient as the dimension grows. 

\begin{remark}
Using the complexity analysis in the previous sections, a similar comparison analysis is straightforward for the $p$-Laplace equation with $p=2$ as well as for the HJB and the Heat equations.
\end{remark}
\subsection{Expressive power of smoothed networks}
\label{sec:ExpreSmooth}
From Theorem~\ref{thm:lip}, 
when the output of the quantum network is Gaussian smoothed, 
the resulting PDE trial solution is $\frac{\ub}{\sigma}\sqrt{\frac{2}{\pi}}$ Lipschitz, where $\ub = \sup_{x \in \RR^d}  |u_\Theta(x)|$ with respect to the $2$-norm. 
To derive an upper bound for~$\ub$ for any parameter set it suffices to consider the range of the expectation for the Hermitian cost \eqref{fullquantnetwork} operator
\begin{equation}
    \sup_{x,\Theta}  |u_\Theta(x)|
     = \sup_{x,\Theta}  |\langle\Cc\rangle_{\Ug_{\Theta}(x) \ket{\boldsymbol{0}}}|
    \leq \sup_{\psi, \braket{\psi | \psi} = 1}\left\lvert \braket{\Cc}_{\psi}\right\rvert,
\end{equation}
using the Ising Hamiltonian~\eqref{eq:IsingHamilt}

\begin{equation}
    \sup_{\psi, \braket{\psi | \psi} = 1}\left\lvert \braket{\Cc}_{\psi}\right\rvert \leq \sup_{\psi, \braket{\psi | \psi} = 1}\left\lvert \left\langle
     \sum_{j=1}^n \left[ \widehat{\Zg}_{j \Mod{n}
}\widehat{\Zg}_{j+1 \Mod{n}} + \widehat{\Zg}_{j \Mod{n}} + \widehat{\Xg}_{j\Mod{n}} \right]\right\rangle_{\psi}\right\rvert \leq 3n,
\end{equation}
resulting in a Lipschitz constant of $\frac{3n}{\sigma}\sqrt{\frac{2}{\pi}}$ for the PDE trial solution, unlikely to be the best Lipschitz constant. For example for a single qubit it is easy to see that the expectation of~$\Zg$ will not be~$1$ when the expectation of~$\Xg$ is~$1$. 


\subsection{Lipschitz constant for the heat equation}
\subsubsection{Heat equation with small Lipschitz constant}
Consider the heat equation defined as
\begin{equation*}
\left\{
 \begin{array}{rll}
 \Delta  u(t,x) & = \dt u(t,x),
  & \text{for all }x \in \Bb_{0,1}, t \in [0,T),\\
  u(0,x) & = \displaystyle \frac{\|x\|^2}{2d},
  & \text{for all } x \in \Bb_{0,1},\\
 u(t,x) & = \displaystyle t + \frac{1}{2d},
  & \text{for all }x \in \partial \Bb_{0,1},  t \in [0,T).
  \end{array}
  \right.
\end{equation*}
with $\Bb_{0,1}$, the unit ball in~$\RR^d$.
The solution is $u(t,x) = t + \frac{\|x\|^2}{2d}$, with Lipschitz constant
$$
L^{\alpha}(u)
= \sup_{x \in \Bb_{0,1}, t \in [0,T) }
\left\|\nabla u (t,x)\right\|_{2}
= \sup_{x \in\Bb_{0,1}} \left\| \frac{x}{d} + 1 \right\|_{2}
= \frac{1}{d} + 1.
$$
In our formulation, $n \leq d$, 
where~$n$ is the number of qubits
and the data lies in~$\RR^d$.
Therefore, the Lipschitz constant~$L^{\alpha}(u)$ is smaller than that of the PDE trial solution, $\frac{3n}{\sigma}\sqrt{\frac{2}{\pi}}$, obtained in Section~\ref{sec:ExpreSmooth} as long as $\sigma\leq 
3n\sqrt{\frac{2}{\pi}}(1+\frac{1}{d})^{-1}$.
Good experimental results with $\sigma \leq 0.1$ were outlined in~\cite{he2023learning},
which clearly satisfies the condition for all integers~$n$ and~$d$.

\subsubsection{Heat equation with large Lipschitz constant}
For an example of PDE where this Gaussian method is intractable with quantum networks, 
consider the heat equation
\begin{equation}\label{eq:heat2}
 \Delta  u(t,x) = \dt u(t,x),
\end{equation}
for $x \in [0,1]^d$ and $t \in [0,T)$, with Dirichlet boundary conditions, 
and consider the solution given by
\begin{equation}\label{eq:heat2sol}
    u(t,x) = d^{\Oneovd} 
    \E^{-a^2\pi^2 d t}\prod_{i=1}^{d} \sin(a\pi x_i),
\end{equation}
for $x=(x_1, \ldots, x_d)$, and $\nabla_{t,x}$ the gradient with respect to both $t \text{ and } x$ we have the Lipschitz constant upper bound
\begin{align}
 L^{\alpha}(u) 
&= \sup_{x \in [0,1]^d, t \in [0,T) } \left\|\nabla_{t,x}u({t,x}) \right\|_{2} \\
&= \sup_{x \in [0,1]^d, t \in [0,T) }\bigg\lVert
-\bm{e}_1 a^2\pi^2 d u(t,x)
+ a\pi d^{\Oneovd}\E^{-a^2\pi^2 d t}\sum_{j=1}^d \bm{e}_{j+1}\cos(a\pi x_j)
\prod_{i=1,i \neq j}^{d} \sin(a\pi x_i) \bigg\rVert_2 \\
&= \sup_{x \in [0,1]^d}
\bigg\lVert
-\bm{e}_1 a^2\pi^2d u(0,x)
+ a\pi d^{\Oneovd}\sum_{j=1}^d \bm{e}_{j+1}\cos(a\pi x_j)
\prod_{i=1,i \neq j}^{d} \sin(a\pi x_i)\bigg\rVert_2\\
&\geq \bigg\lVert
-\bm{e}_1 a^2\pi^2du\left(0,\half\right)
+ a\pi d^{\Oneovd}\sum_{j=1}^d \bm{e}_{j+1}\cos\left( \frac{a\pi}{2}\right)
\prod_{i=1,i \neq j}^{d} \sin\left(\frac{a\pi}{2}\right)\bigg\rVert_2\\
&= \bigg\lVert
-\bm{e}_1 a^2\pi^2d u\left(0,\half\right)
+ a\pi d^{\Oneovd}\sum_{j=1}^d \bm{e}_{j+1}\cos\left( \frac{a\pi}{2}\right)
\ \sin\left(\frac{a\pi}{2}\right)^{d-1}\bigg\rVert_2\\
&=\sqrt{\bigg(a^2\pi^2 du\left(0,\half\right)\bigg)^2
+ 
\left(a\pi\cos\left(\frac{a\pi}{2}\right)d^{\Oneovd}d\sin\left(\frac{a\pi}{2}\right)^{d-1}\right)^2},
\end{align}
where~$\bm{e}_{j}$ denotes the unit vector in $\RR^d$ with~$1$ for component~$j$ and~$0$ otherwise.
With $d = 50$ and $a = 1$, the resulting Lipschitz constant is greater than $22360$. Comparing this to the Lipschitz restriction found earlier, if we use a value of $\sigma = 0.1$ we would need more than~$900$ qubits for the smoothed PDE trial solution.

\section{Numerical examples}\label{sec:Numerics}
The quantum state simulation is performed using the \texttt{Yao} package for \texttt{Julia}~\cite{Luo_2020}.
We use the random quantum network introduced in Section~\ref{sec:specificnetworks} with a varying number of qubits, which we compare to the random classical network developed by Gonon~\cite{gonon2023random}: 
let~$E_1$ be $t_5(0, \Id_d)$ (multivariate t-distribution) random variable 
and~$B_1$ a Student $t$-distribution with two degrees of freedom. 
At each training iteration we uniformly sample new points from~$\Omega$ and~$\partial\Omega$. 
We train solutions using both classical and quantum networks. 
Due to the random nature of the networks we repeat each training process five times, 
and all training statistics reported below are mean values.

\subsection{Poisson equation}
Consider the Poisson equation~\eqref{eq:poiss} with $p=2$, $\Omega = (0,1)^2$ and
$$
f(x,y) = 26\cos(y)\cos(5x) -\frac{2y}{(1 + x)^3},
$$
so that the explicit solution reads
$$
u(x,y) = \cos(5x)\cos(y)+\frac{y}{1+x}+\sin(x+y)^2+\cos(x+y)^2.
$$
Alongside quantum and classical networks, 
we also investigate the two loss functions~$\mathcal{L}^{\text{Var}}$ and~$\mathcal{L}$. 
To demonstrate the effectiveness of the Haar random operator we also train solutions using $\bm{\Lambda} = \Id$.
We use six qubits, detailed training information and network settings are shown in Table~\ref{tab:training1}.
Final relative~$L_2$ errors of the trial solutions are shown in Table~\ref{tab:my-table} alongside other metrics. Regardless of the loss function used the full random quantum networks outperform all of the random classical networks. 
In Table~\ref{tab:my-table} we see that the performance of the classical networks improves when the number of nodes increases, notably the full random quantum network has $24$ trainable parameters yet it outperforms the classical network with more than $4$ times the number of trainable parameters. 
We also see that for both classical and quantum networks the variational formulation loss functions produces trial solutions with approximately the same final~$L_2$ relative error.

The addition of the Haar random operator~$\bm{\Lambda}$ has little impact on the training complexity since it is fixed.
However, it greatly improves the final~$L_2$ relative error, 
for example reducing the error for the variational loss from~$0.575$ to~$0.040$. 
Figure~\ref{fig:foobar11} shows the squared error $|u_\Theta -u|^2$ over the domain of~$x_1$ for $x_2 \in\{0.25,0.5,0.75\}$.
The solid blue line indicates the average value for the quantum network with the shaded blue area 
representing all error values from the five runs. 
For the classical networks we plot the network with the lowest overall error of the five. 
Figure~\ref{fig:foobar11} shows snapshots of the trial solution against the analytic solution. 
The five quantum networks display similar behaviour over the intervals when compared to each other.
\begin{figure}
\centering
\includegraphics[angle=-90,scale=0.4]{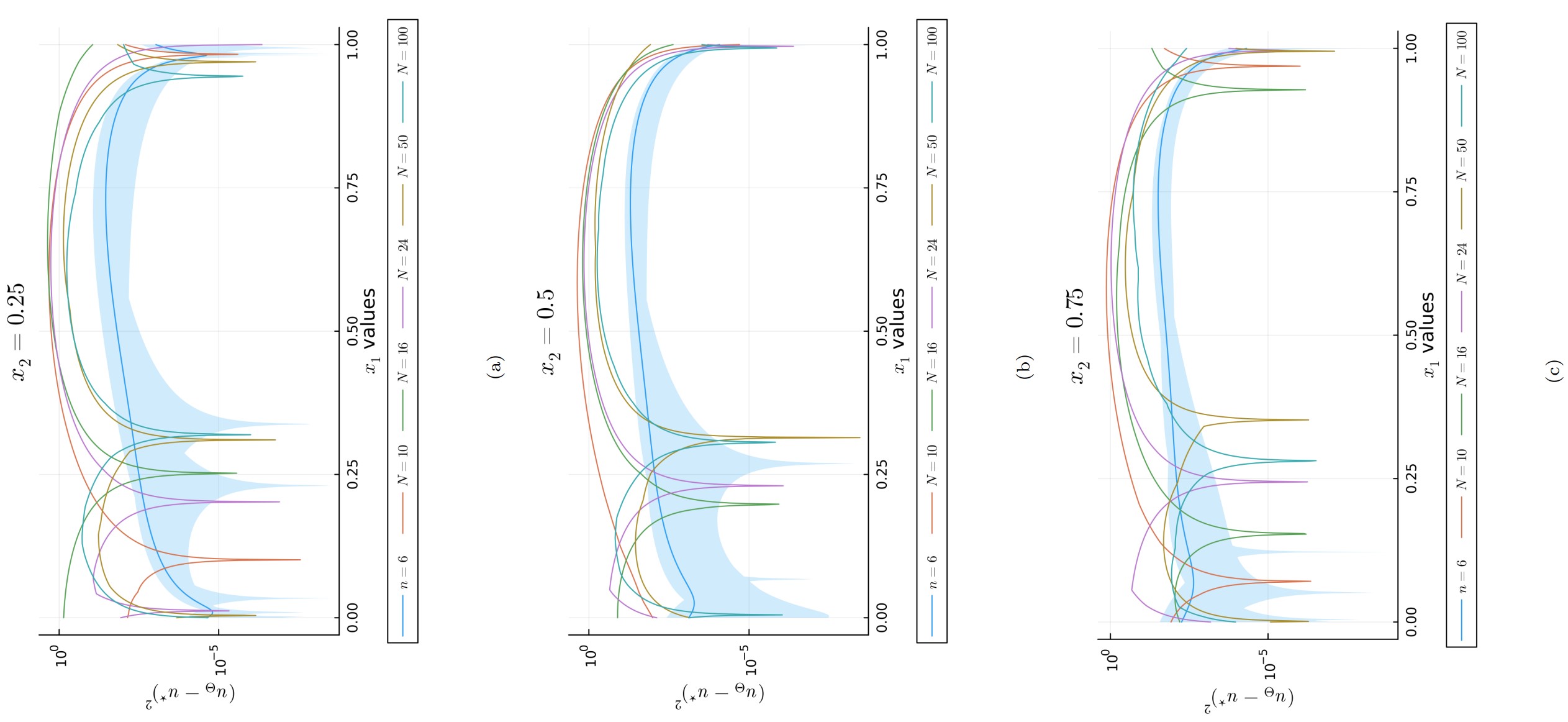}
 \caption{Squared error values $|u_\Theta(x) -u(x)|^2$  of the trial solutions for the Poisson equation in~$\RR^2$ with~$x_1$ on the horizontal axis and for $x_2 \in\{0.25,0.5,0.75\}$.
    The parameter~$N$ refers to the number of nodes in the random classical network and $n=6$ is the number of qubits in the random quantum network. 
    For the classical networks we plot the solution with the lowest overall mean squared error whereas for the quantum networks we plot the average squared error with a ribbon indicating the range of error for the five different training runs.}
        \label{fig:foobar11}
\end{figure}

\begin{figure}
\centering
\includegraphics[angle=-90,scale=0.4]{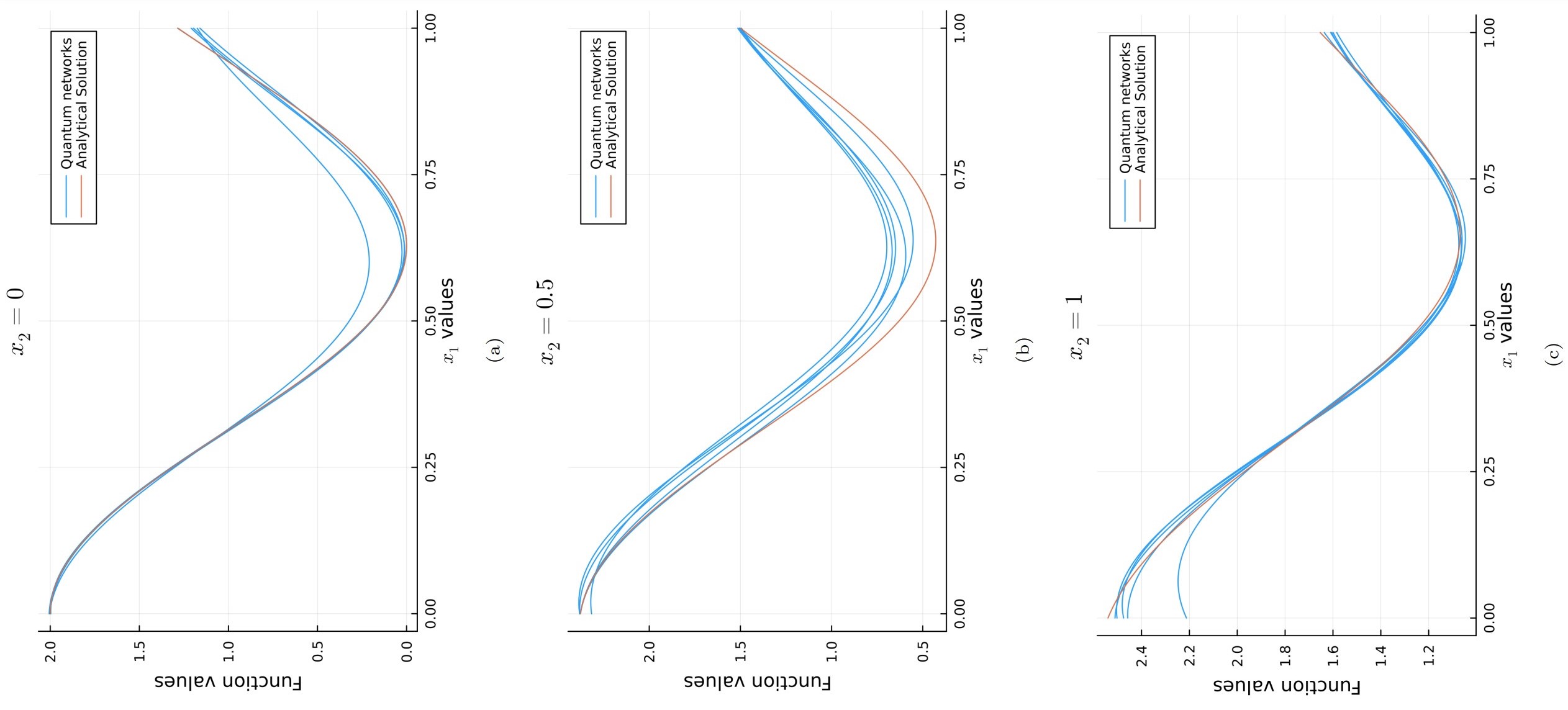}
    \caption{Results of the quantum networks for the Poisson equation in~$\RR^2$,  as a function of~$x_1$ for $x_2\in\{0,0.5,1\}$.
    We also plot the analytic solution for comparison.}
    \label{fig:foobar22}
\end{figure}

\begin{table}[]
\centering
\begin{tabular}{@{}ll|cc|@{}}
\cmidrule(l){3-4}
&  & \multicolumn{2}{c|}{$L_2$ Relative Error}\\ \cmidrule(l){3-4} 
&  & \multicolumn{1}{l|}{Using $\mathcal{L}^{\text{Var}}$} & \multicolumn{1}{l|}{Using $\mathcal{L}$} \\ \midrule
\multicolumn{1}{|c|}{\multirow{2}{*}{Quantum Networks}}   & Full random network & \multicolumn{1}{c|}{$0.040 \pm 0.041$}                 & $0.051 \pm 0.067$                        \\ \cmidrule(l){2-4} 
\multicolumn{1}{|c|}{}                                    & $\bm{\Lambda} = \Id$  & \multicolumn{1}{c|}{$0.575 \pm 0.27$}                 & $0.61 \pm 0.22$                         \\ \midrule
\multicolumn{1}{|l|}{\multirow{4}{*}{Classical Networks}} & $N = 10$            & \multicolumn{1}{c|}{$0.683 \pm 1.7$}                 & $0.799 \pm 1.30$                        \\ \cmidrule(l){2-4} 
\multicolumn{1}{|l|}{}                                    & $N = 16$            & \multicolumn{1}{c|}{$0.712 \pm 0.43$}                 & $0.616 \pm 0.59$                        \\ \cmidrule(l){2-4} 
\multicolumn{1}{|l|}{}                                    & $N = 24$           & \multicolumn{1}{c|}{$0.511 \pm 0.19$}                 & $0.537 \pm 0.13$                        \\ \cmidrule(l){2-4} 
\multicolumn{1}{|l|}{}                                    & $N = 50$           & \multicolumn{1}{c|}{$0.281 \pm 0.093$}                 & $0.200 \pm 0.081$  
\\ \cmidrule(l){2-4} 
\multicolumn{1}{|l|}{}                                    & $N = 100$           & \multicolumn{1}{c|}{$0.239 \pm 0.11$}                 & $0.156 \pm 0.074$                
\\ \bottomrule
\end{tabular}
\label{tab:my-table}
\caption{Training results for the random networks used to solve the Poisson equation with varying loss function formulations.  }
\end{table}

\subsection{Heat equation}
We consider the heat equation~\eqref{eq:heat2} with the solution~\eqref{eq:heat2sol} and $T = 1$, $a=0.25$. We solve with $d=1,2$ using $4$ and $6$ qubits, respectively. 
Specific training settings are shown in Appendix~\ref{tab:training2} and detailed training results can be found in Table~\ref{tab:heat1} and Table~\ref{tab:heat2}.
For $d=1$, the quantum network has a lower average MSE than the average values for the classical networks. 
Specifically, the classical network with the same number of trainable parameters has an MSE an order of magnitude larger.  We also train classical networks with both more and fewer trainable parameters and see that the quantum networks outperform all the classical ones.
For $d=2$, the quantum network has~$60$ trainable parameters and outperforms all the random classical networks with less than~$70$ parameters. 
There is a large boost in final MSE when the classical random network has more than~$80$ nodes;
compared to the previous two examples, this could suggest that our quantum network lacks the expressivity needed or more training iterations are required for accurate quantum networks.
In Figure~\ref{fig:foobar1}, we plot the squared error values $|u_\Theta -u|^2$ over the domain of~$x_1$ and at values of $x_2 \in\{0.25,0.5,0.75\}$.
The solid blue line indicates the average value for the quantum network with the shaded blue area being a ribbon that covers all error values from the five runs. For the classical networks we plot the network with the lowest overall error of the five. 
In Figure~\ref{fig:foobar2} we plot snapshots of the trial solution against the analytic solution. 
\begin{figure}
\centering
\includegraphics[angle=-90,scale=0.4]{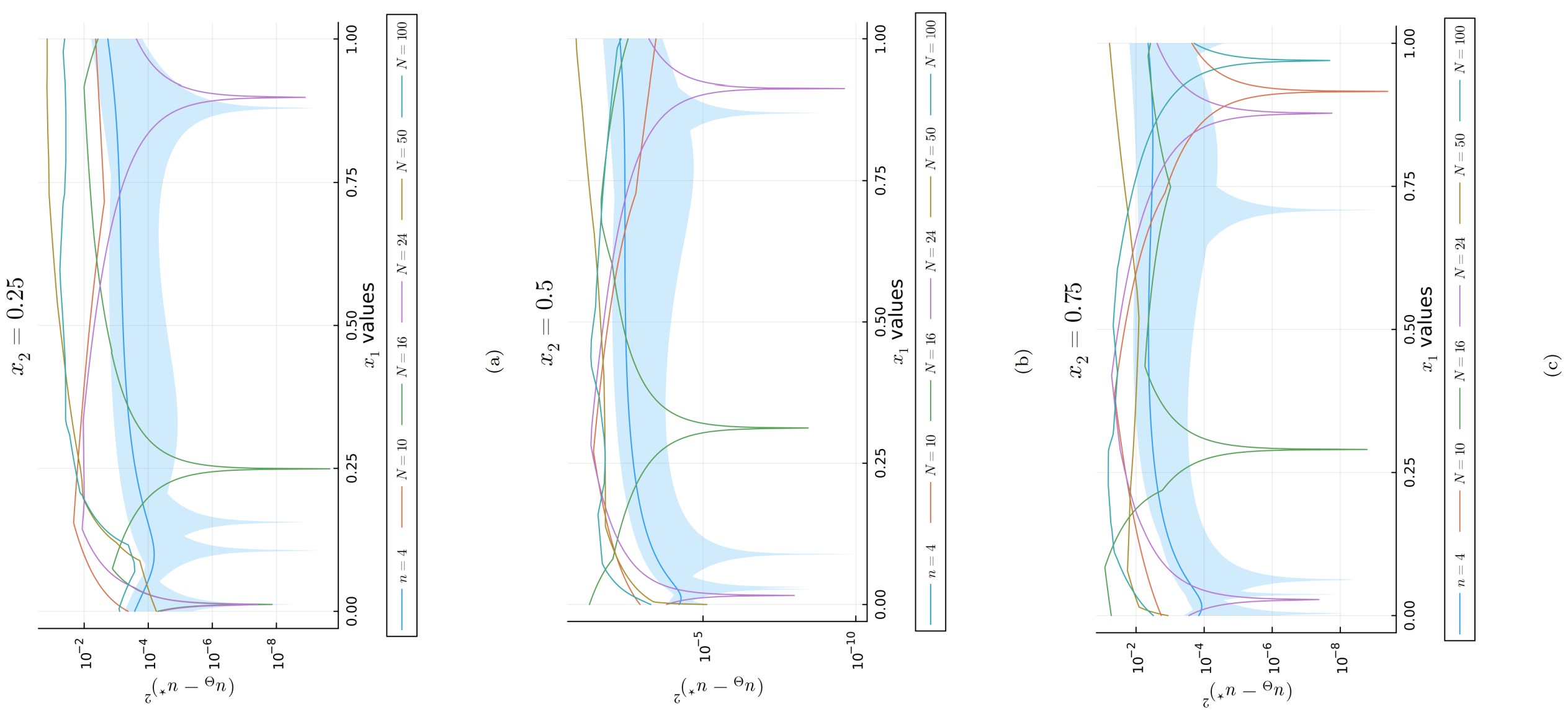}
    \caption{Squared error values $|u_\Theta(x) -u(x)|^2$  of the trial solutions for the Heat equation with $d=1$ and~$x_1$ on the horizontal axis and for $x_2 \in\{0.25,0.5,0.75\}$.
    The parameter~$N$ refers to the number of nodes in the random classical network and $n=4$ is the number of qubits in the random quantum network. 
    For the classical networks we plot the solution with the lowest overall mean squared error whereas for the quantum networks we plot the average squared error with a ribbon indicating the range of error for the five different training runs.}
    \label{fig:foobar1}
\end{figure}

\begin{figure}
\centering
\includegraphics[angle=-90,scale=0.4]{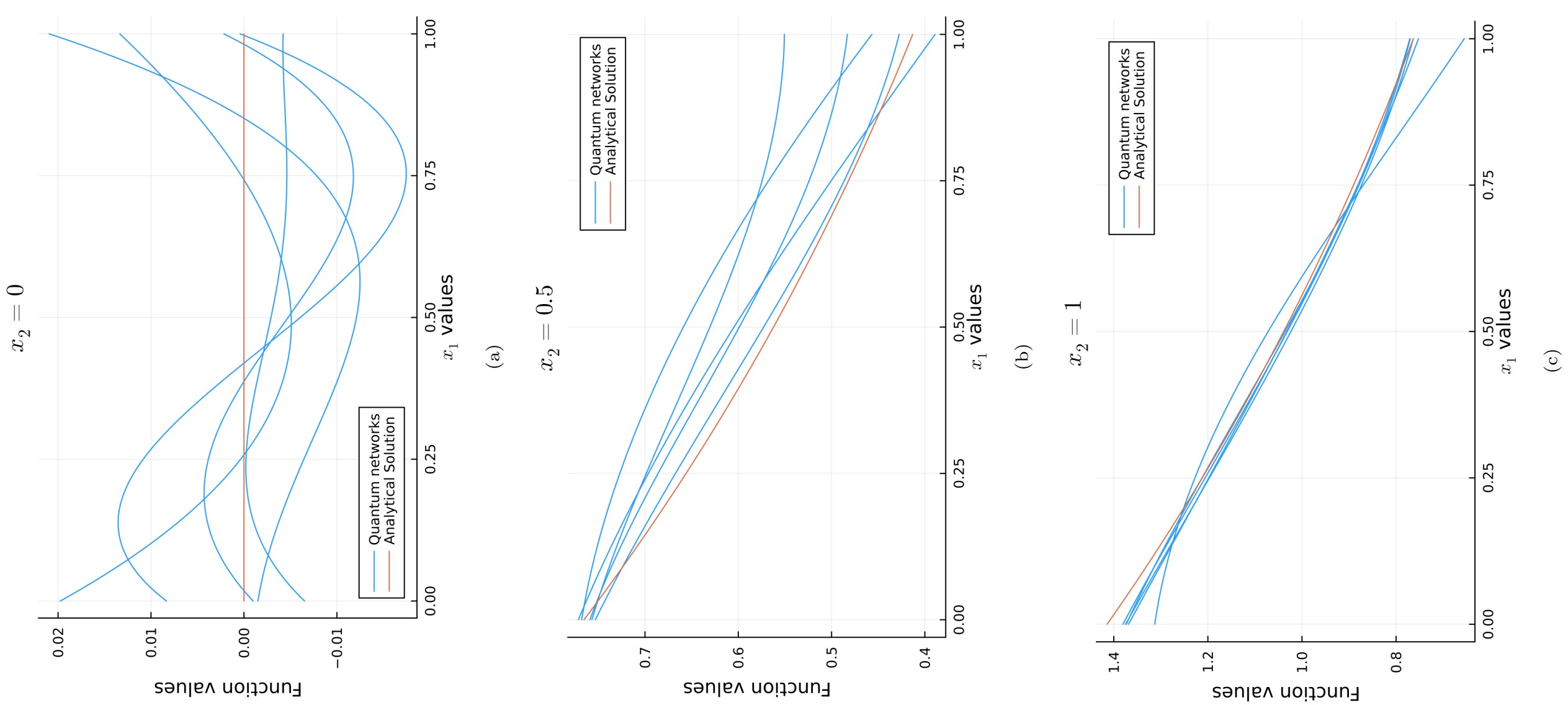}
    \caption{Solution snapshots of the quantum networks for the Heat equation with $d=1$ and~$x_1$ on the horizontal axis and $x_2 \in \{0,0.5,1\}$.
     We also plot the analytic solution for comparison.}
    \label{fig:foobar2}
\end{figure}

\begin{table}[]
\parbox{.47\linewidth}{
\centering
\begin{tabular}{@{}lc|c|@{}}
\cmidrule(l){3-3}
 & \multicolumn{1}{l|}{} & MSE Error \\ \midrule
\multicolumn{1}{|c|}{Quantum}                     & $n=4$                 & $(7.8 \pm 3.5) \textrm{E-4} $ \\ \midrule
\multicolumn{1}{|l|}{\multirow{5}{*}{Classical}} & $N=10$                & $(1 \pm 5.3) \textrm{E-2}$  \\ \cmidrule(l){2-3} 
\multicolumn{1}{|l|}{}                                    & $N = 16$              & $(8 \pm 5.8) \textrm{E-3}$  \\ \cmidrule(l){2-3} 
\multicolumn{1}{|l|}{}                                    & $N = 24$              & $(1.12 \pm 0.87) \textrm{E-2}$  \\ \cmidrule(l){2-3} 
\multicolumn{1}{|l|}{}                                    & $N = 50$              & $(4.1 \pm 0.65) \textrm{E-2}$  \\ \cmidrule(l){2-3} 
\multicolumn{1}{|l|}{}                                    & $N = 100$             & $(2.7 \pm 1.2) \textrm{E-2}$  \\ \bottomrule
\end{tabular}
\caption{Results for the random networks for the Heat equation with $d=1$.}
\label{tab:heat1}
}
\hfill
\parbox{.47\linewidth}{
\begin{tabular}{@{}lc|c|@{}}
\cmidrule(l){3-3}
& \multicolumn{1}{l|}{} & MSE Error                        \\ \midrule
\multicolumn{1}{|c|}{Quantum}                     & $n=6$                 & $0.06 \pm 0.04$ \\ \midrule
\multicolumn{1}{|l|}{\multirow{5}{*}{Classical}} & $N=30$                & $0.36 \pm 0.21$  \\ \cmidrule(l){2-3} 
\multicolumn{1}{|l|}{}                                    & $N = 40$              & $0.23 \pm 0.18 $  \\ \cmidrule(l){2-3} 
\multicolumn{1}{|l|}{}                                    & $N = 50$              & $0.17 \pm 0.05$  \\ \cmidrule(l){2-3} 
\multicolumn{1}{|l|}{}                                    & $N = 60$              & $0.09 \pm 0.11 $  \\ \cmidrule(l){2-3} 
\multicolumn{1}{|l|}{}                                    & $N = 70$             & $0.07 \pm 0.09$  \\ \cmidrule(l){2-3} 
\multicolumn{1}{|l|}{}                                    & $N = 80$             & $0.025 \pm 0.05$  \\ \cmidrule(l){2-3} 
\multicolumn{1}{|l|}{}                                    & $N = 100$             & $0.038 \pm 0.08$  \\ \bottomrule
\end{tabular}
\caption{Results for the random networks for the Heat equation with $d=2$.}
\label{tab:heat2}
}
\end{table}

\subsection{Hamilton–Jacobi equation}
We use the specific HJB equation~\eqref{eq:HJB}
with the unique solution~\cite{Han_2018}
\begin{equation}
u(t,x) = -\mu(2\pi)^{-\frac{{d}}{2}}
\log\left(\int_{\RR^{d}}
\exp\left\{-\frac{\|y\|^2}{2} - \mu h\left(x-\sqrt{2(T-t)}y\right)\right\}\D y\right).
\end{equation}
Due to the domain of $x$ being $\RR^{d}$ we require a different form of the encoding matrix, as a result we change definition~\eqref{eq:rot} to
\begin{equation}
    \bm{A}(x) := \bigotimes_{j = 1 }^{n}
\Rg_{z}\left[\mathcal{X}_j\pi\tanh\left(x_{\mathfrak{g}_{j}}\right)\right].
\end{equation}
We use the activation function $\tanh$ due to it's performance in traditional machine learning applications. 
We solve the specific HJB equation~\eqref{eq:HJB} with $d = 2$, $\mu = 1$, $T=1$,
$h(x) = \log\left( \frac{1+\|x\|^2}{2}\right)$
and~$9$ qubits alongside~$1$ trainable layer. 
Training results are shown in Table~\ref{tab:hjb} and training settings in Appendix~\ref{tab:training3hjb}. 
We see relatively poor performance for both sets of random networks, which is due to hardware limitations.
Indeed, the MSE during training did not plateau for any of the models, showing than more training iterations should be used.
Calculating the derivatives needed for the HJB equation using a network architecture of~$9$ qubits requires much larger computational resources, 
which we leave to further study.
\begin{table}[]
\centering
\begin{tabular}{@{}lc|c|@{}}
\cmidrule(l){3-3}
                                                          & \multicolumn{1}{l|}{} & MSE Error     \\ \midrule
\multicolumn{1}{|c|}{Quantum Network}                     & $n=9$                & $0.11 \pm 0.13$ \\ \midrule
\multicolumn{1}{|l|}{\multirow{2}{*}{Classical Networks}} & $N=50$                & $0.31 \pm 0.19$  \\ \cmidrule(l){2-3} 
\multicolumn{1}{|l|}{}                                    & $N = 75$              & $0.27 \pm 0.13$  \\ \bottomrule
\end{tabular}
\caption{Training results for the random networks used to solve the given HJB equation.}
\label{tab:hjb}
\end{table}


\section{Conclusion}
This paper develops parameterised quantum circuits to solve widely used PDEs in any dimension.
It further provides a precise complexity study of these algorithms and compare them to their classical counterparts, highlighting their potential advantages and limitations.
\backmatter

\bmhead{Acknowledgements} 
The authors are grateful to the Department of Aeronautics at Imperial College London for supporting this work with a doctoral studentship.
SL gratefully acknowledges financial support from the EPSRC grant EP/W032643/1
and AJ that of the EPSRC grants EP/W032643/1 and EP/T032146/1.
\textit{‘For the purpose of open access, the author has applied a ‘Creative Commons Attribution (CC BY) licence to any Author Accepted Manuscript version arising’}
\section*{Declarations}
The authors have no relevant financial or non-financial interests to disclose. The authors have no conflicts of interest to declare that are relevant to the content of this article. All authors certify that they have no affiliations with or involvement in any organization or entity with any financial interest or non-financial interest in the subject matter or materials discussed in this manuscript. The authors have no financial or proprietary interests in any material discussed in this article.

\begin{appendix}
\section{Poisson equation training details}
We use iid sample points $\left\{x^{(i)} \right\}_{i=1,\ldots, n_r}$ drawn uniformly in $(0,1)^{d}$, and $n_e$ iid sample points $\left\{\tilde{x}^{(i)}\right\}_{i=1,\ldots, n_e}$
drawn uniformly from the boundary $\partial(0,1)^{d}$. For the $L_2$ relative error statistics we use $1000$ sample points of the form $\left\{x^{(i)} \right\}_{i=1,\ldots, 1000}$.
\begin{table}[!htbp]\centering \label{tab:training1}
$$
\begin{array}{lc}
\text {Specific Training Details for Poisson equation.}\\
\hline \text { Quantum Model Configuration } & \\
\hline \text { Trainable Layers } & 1 \\
\text { Number of Qubits } & 6 \\
\text { Trainable Parameters } & 24 \\
\hline \text { Classical Model Configuration } & \\\hline \text { Nodes / Trainable Parameters } &  {10,16,24,50,100} \\
\hline \text { Hyperparameters } & \\
\hline \text { Total iterations } & 1000 \\
\text { Domain Batch Size } n_r & 128 \\
\text { Boundary Batch Size } n_e & 64 \\
\text { Boundary Weighting } \lambda_e & 1 \\
\text { Optimiser } & \text{Adam} \\
\text { Learning Rate } & \text{1E-3} \\
\text { Adam } \varepsilon & \text{1E-8} \\
\text { Adam }\left(\beta_1, \beta_2\right) & (0.9,0.999) \\
\hline
\end{array}
$$
\end{table}

\section{Heat equation training settings}
We use iid sample points $\left\{(t^{(i)}, x^{(i)}) \right\}_{i=1,\ldots, n_r}$ drawn uniformly in $[0,1]\times (0,1)^{d}$ and $n_e$ iid sample points $\left\{\tilde{x}^{(i)}\right\}_{i=1,\ldots, n_e}$
drawn uniformly from $\partial(0,1)^{d}$. For the MSE statistics we use $1000$ sample points of the form $\left\{(t^{(i)}, x^{(i)}) \right\}_{i=1,\ldots, 1000}$.
\begin{table}[!htbp]\centering \label{tab:training2}
$$
\begin{array}{lc}
\text {Specific Training Details for the Heat equation in $2$ dimensions.}\\
\hline \text { Quantum Model Configuration } & \\
\hline \text { Trainable Layers } & 1 \\
\text { Number of Qubits } & 4 \\
\text { Trainable Parameters } & 16 \\
\hline \text { Classical Model Configuration } & \\\hline \text { Nodes / Trainable Parameters } &  {10,16,24,50,100} \\
\hline \text { Hyperparameters } & \\
\hline \text { Total iterations } & 1000 \\
\text { Domain Batch Size } n_r & 128 \\
\text { Boundary Batch Size } n_e & 64 \\
\text { Boundary Weighting } \lambda_e & 500 \\
\text { Gradient Clipping } & \pm 1 \\
\text { Optimiser } & \text{Adam} \\
\text { Learning Rate } & 5E-3 \\
\text { Adam } \varepsilon & \text{1E-8} \\
\text { Adam }\left(\beta_1, \beta_2\right) & (0.9,0.999) \\
\hline
\end{array}
$$
\end{table}
\begin{table}[!htbp]\centering \label{tab:training3}
$$
\begin{array}{lc}
\text {Specific Training Details for the Heat equation in $3$ dimensions.}\\
\hline \text { Quantum Model Configuration } & \\
\hline \text { Trainable Layers } & 2 \\
\text { Number of Qubits } & 6 \\
\text { Trainable Parameters } & 60 \\
\hline \text { Classical Model Configuration } & \\\hline \text { Nodes / Trainable Parameters } &  {30,40,50,60,70,80,100} \\
\hline \text { Hyperparameters } & \\
\hline \text { Total iterations } & 2000 \\
\text { Domain Batch Size } n_r & 64 \\
\text { Boundary Batch Size } n_e & 64 \\
\text { Boundary Weighting } \lambda_e & 500 \\
\text { Gradient Clipping } & \pm 1 \\
\text { Optimiser } & \text{Adam} \\
\text { Learning Rate } & 5E-3 \\
\text { Adam } \varepsilon & \text{1E-8}\\
\text { Adam }\left(\beta_1, \beta_2\right) & (0.9,0.999) \\
\hline
\end{array}
$$
\end{table}

\section{HJB equation training settings}
We use iid sample points $\left\{(t^{(i)}, x^{(i)}) \right\}_{i=1,\ldots, n_r}$
and $\left\{\tilde{x}^{(i)}\right\}_{i=1,\ldots, n_e}$,
where $\left\{t^{(i)}\right\}_{i=1,\ldots, n_r}$
are drawn uniformly in $[0,1]$,
and 
$\left\{x^{(i)}\right\}_{i=1,\ldots, n_r}$
and 
$\left\{\tilde{x}^{(i)}\right\}_{i=1,\ldots, n_e}$
are drawn from centered normalised Gaussian distributions in~$\RR^d$. For the MSE statistics we use $1000$ sample points of the form $\left\{(t^{(i)}, x^{(i)}) \right\}_{i=1,\ldots, 1000}$.
\begin{table}[!htbp]\centering \label{tab:training3hjb}
$$
\begin{array}{lc}
\text {Specific Training Details for the HJB equation in $3$ dimensions.}\\
\hline \text { Quantum Model Configuration } & \\
\hline \text { Trainable Layers } & 1 \\
\text { Number of Qubits } & 9 \\
\text { Trainable Parameters } & 45 \\
\hline \text { Classical Model Configuration } & \\\hline \text { Nodes / Trainable Parameters } &  {50,75} \\
\hline \text { Hyperparameters } & \\
\hline \text { Total iterations } & 750 \\
\text { Domain Batch Size } n_r & 64 \\
\text { Boundary Batch Size } n_e & 64 \\
\text { Boundary Weighting } \lambda_e & 500 \\
\text { Gradient Clipping } & \pm 1 \\
\text { Optimiser } & \text{Adam} \\
\text { Learning Rate } & \text{5E-3}\\
\text { Adam } \varepsilon & \text{1E-8} \\
\text { Adam }\left(\beta_1, \beta_2\right) & (0.9,0.999) \\
\hline
\end{array}
$$
\end{table}
\end{appendix}

\newpage
\clearpage
\bibliography{sn-bibliography}

\end{document}